\documentclass{article}

\usepackage{amsmath}
\usepackage{amsthm}
\usepackage{appendix}
\usepackage{bbold}

\usepackage{tikz}
\usepackage{pgfplots}
\usepackage{pgf}
\usetikzlibrary{patterns}
\usepackage{float}
\usepackage{graphicx}
\usepackage{xcolor}
\usepackage{hyperref}

 \theoremstyle{plain}
  \newtheorem{thm}{Theorem}
  
  \newtheorem{lemma}[thm]{Lemma}
  \newtheorem{cor}[thm]{Corollary}
  
  \newtheorem*{thm*}{Theorem}
  \newtheorem*{prop*}{Proposition}
  \newtheorem*{lemma*}{Lemma}
  \newtheorem*{cor*}{Corollary}
  \newtheorem*{remark*}{Remark}

\theoremstyle{definition}
\newtheorem{definition}[thm]{Definition}
 \newtheorem*{conj*}{Conjecture}

\usepackage{fullpage}

\usepackage{todonotes}

\newtheorem{innercustomthm}{Lemma}
\newenvironment{customlemma}[1]
  {\renewcommand\theinnercustomthm{#1}\innercustomthm}
  {\endinnercustomthm}

\usepackage{authblk}

\def\norm #1{\Vert #1\Vert}

\def\bra #1{\langle #1\vert}
\def\ket #1{\vert #1\rangle}

\def\ketbra #1#2{\vert #1\rangle \langle #2\vert}
\def\kettbra#1{\ketbra{#1}{#1}}
\def\tr{{\rm Tr}}

\def\P{{\mathbb P}}

\newcommand*{\eps}{\varepsilon}

\newcommand*{\half}{\frac{1}{2}}

\newcommand*{\Ext}{\mathrm{Ext}}

\usepackage{ifthen}

\newcommand{\proj}[1]{\ket{#1}\bra{#1}}
\newcommand{\HminOp}{H_{\text{min}}}

\newcommand{\sHmin}[2][@]{\ensuremath{\HminOp(#2)\ifthenelse{\equal{#1}{@}}{}{_{#1}}}}
\newcommand{\lHmin}[2][@]{\ensuremath{\HminOp\left(#2\right)\ifthenelse{\equal{#1}{@}}{}{_{#1}}}}
\newcommand{\Hmin}[2][@]{\if@display\lHmin[#1]{#2}\else\sHmin[#1]{#2}\fi}

\newcommand{\sHminSmooth}[3][@]{\ensuremath{\HminOp^{#2}(#3)\ifthenelse{\equal{#1}{@}}{}{_{#1}}}}
\newcommand{\lHminSmooth}[3][@]{\ensuremath{\HminOp^{#2}\left(#3\right)\ifthenelse{\equal{#1}{@}}{}{_{#1}}}}
\newcommand{\HminSmooth}[3][@]{\if@display\lHminSmooth[#1]{#2}{#3}\else\sHminSmooth[#1]{#2}{#3}\fi}

\newcommand{\snorm}[1]{\ensuremath{\|#1\|}}
\newcommand{\lnorm}[1]{\ensuremath{\left\|#1\right\|}}
\newcommand{\trnorm}[1]{\if@display\lnorm{#1}\else\snorm{#1}\fi}

\begin{document}

\title{Quantum-proof multi-source randomness extractors \\ in the Markov model}

\author[1]{Rotem Arnon-Friedman}
\author[1]{Christopher Portmann}
\author[1,2]{Volkher B. Scholz}
\affil[1]{Institute for Theoretical Physics, ETH-Z\"urich, CH-8093, Z\"urich, Switzerland}
\affil[2]{Department of Physics and Astronomy, Ghent University, Krijgslaan 281, S9, B-9000 Ghent, Belgium}

\date{\today}

\maketitle

\begin{abstract}
	Randomness extractors, widely used in classical and quantum cryptography and other fields of computer science, e.g., derandomization, are functions which generate almost uniform randomness from weak sources of randomness. In the quantum setting one must take into account the quantum side information held by an adversary which might be used to break the security of the extractor.  
	In the case of seeded extractors the presence of quantum side information has been extensively studied. For multi-source extractors one can easily see that high conditional min-entropy is not sufficient to guarantee security against arbitrary side information, even in the classical case.
	Hence, the interesting question is under which models of (both quantum and classical) side information multi-source extractors remain secure. 
	In this work we suggest a natural model of side information, which we call the Markov model, and prove that any multi-source extractor remains secure in the presence of quantum side information of this type (albeit with weaker parameters). This improves on previous results in which more restricted models were considered or the security of only some types of extractors was shown.
\end{abstract}

\section{Introduction}\label{sec:intro}

Randomness extractors are of great importance in many applications in computer science, such as derandomization and cryptography.  The goal of a randomness extractor is to generate (almost) uniform randomness from weak sources of randomness. A weak source is usually modelled as a distribution $X$ over $\{0,1\}^n$ such that the min-entropy of $X$ is lower bounded by $k$: $H_{\text{min}}(X) \geq k$. That is, the source is defined via a probability distribution for which the probability of any string $x\in\{0,1\}^n$ is at most $2^{-k}$. The idea is then to apply a randomness extractor to the weak source, such that the output source $Y$ is indistinguishable from a uniformly random source. 

Unfortunately, no deterministic function can extract the randomness from all sources with a given min-entropy, even for sources with high min-entropy~\cite{santha1986generating,shaltiel2002recent}. The most common ways to avoid this problem are to consider seeded extractors and multi-source extractors. In the case of seeded extractors one uses an additional truly uniform (but short) and independent seed, together with the weak source, as the input to the extractor (see, e.g.,~\cite{impagliazzo1989pseudo,trevisan2001extractors,shaltiel2002recent}).

Alternatively, and of special importance in applications where a uniform seed is not available (e.g., in quantum randomness amplification protocols, see Section~\ref{sec:randomness_amplification}), multi-source randomness extractors can be used. In the multi-source case, instead of starting with one weak source $X$, one considers several \emph{independent} weak sources $X_1,X_2,\dotsc,X_l$ for some $l\geq 2$, with $H_{\text{min}}(X_i) \geq k_i$ for $i\in[l]$, as the input to the extractor (see, for example,~\cite{vazirani1987strong,dodis2004improved,raz05extractors, rao2009extractors,chattopadhyay2015explicit}). 

In all types of extractors the randomness present in the weak sources must be lower bounded for the extractor to work (i.e., a bound on the min-entropy is given as a promise). However, this randomness inherently depends on the information one has about the weak sources, or to put differently, on the \emph{side information} about the sources. For example, extractors are widely used for privacy amplification in cryptographic tasks. There, the starting point is that an adversary holds some side information $C$ about the source such that the \emph{conditional} min-entropy is bounded: $H_{\text{min}}(X|C)\geq k$. The extractor is then used to transform $X$ to a key $Y$, which should be close to uniform even conditioned on the side information $C$. If the extractor fulfils this requirement it is said to be secure. 

Depending on the application one can consider adversaries with classical or quantum side information and ask whether an extractor remains secure even in the presence of such side information (with slightly weaker parameters). For seeded extractors this question has been extensively studied. In the presence of classical adversaries the side information about $X$ can be translated to a decrease in the min-entropy and the extractor remains secure~\cite{konig2008bounded}. In the quantum case, it was further shown in~\cite{konig2008bounded} that all one-bit output extractors remain secure. It is still unknown whether all multi-bit output extractors remain secure (although the results of~\cite{berta2014quantum,quantumbilinear} go in this direction\footnote{Note that there is no contradiction between the results of~\cite{berta2014quantum,quantumbilinear} and the famous counter example of a seeded extractor which breaks in the presence of quantum side information given in~\cite{gavinsky2007exponential}; for details see~\cite{berta2014quantum,quantumbilinear}.}), but several constructions of seeded extractors with good parameters were shown to work also in the presence of quantum side information~\cite{renner2005universally,tomamichel2011extractor,de2012trevisan,hayashi2013dualhashing}.

When considering multi-source extractors things get more complicated, even in the classical case. To see this, consider any one-bit output two-source extractor and let the adversary hold as side information the output of the extractor $Y=\mathrm{Ext}\left( X_1,X_2 \right)$. As this is just one bit, $H_{\text{min}}\left(X_1|Y\right)\geq k_1-1$ and  $H_{\text{min}}\left(X_2|Y\right)\geq k_2 -1$. Furthermore, as the sources are independent even $H_{\text{min}}\left(X_1|YX_2\right)$ and $H_{\text{min}}\left(X_2|YX_1\right)$ remain high. Nevertheless, the extractor obviously fails to produce an output which looks uniform given the side information. 
In~\cite{kasher2010two} several more examples are given in which a small amount of classical side information breaks the extractor completely.
 
This implies that one cannot expect to have multi-source extractors which are secure against any classical or quantum side information and thus raises the question: 
\emph{under which assumptions on the structure of the sources and the side information $X_1\dotsi X_lC$ do multi-source extractors remain secure even in the presence of~$C$?} The main objective of this work is to answer this question. In particular, we define a natural condition on the sources and the side information for which \emph{all} multi-source extractors remain secure in the presence of both classical and quantum side information, but with an increase in the error of the extractor---the distance from uniform of the output.

\subsection{Results and contributions}

Our first contribution is a new definition of a quantum-proof multi-source extractor, which is simpler than previous proposals~\cite{kasher2010two,chung2014multi} and yet sufficient to extract from these models. The original classical extractor definition requires the sources to be independent, i.e., in the two-source case one must have $I(X_1:X_2) = 0$, where $I(\cdot:\cdot)$ denotes the mutual information. If an adversary is present and holds some side information $C$, the definition we introduce requires that the two sources be independent from the point of view of this adversary, i.e., $I(X_1:X_2|C) = 0$. This definition is valid for both classical and quantum side information $C$. This means that the sources and the side information should form a \emph{Markov chain} $X_1 \leftrightarrow C \leftrightarrow X_2$. For the case of more than two sources a similar Markov-type condition can be defined and we say that the sources and the side information are in the \emph{Markov model}. The formal definitions are given in Section~\ref{sec:markov_chain_extractors}.

Compared to previous definitions of quantum-proof multi-source extractors, this has several advantages. Firstly, it is a natural generalization of the original classical extractor definition and the extension to quantum side information from \cite{kasher2010two}, and it connects to the model of \cite{chung2014multi} in the following sense: any function satisfying our definition of a strong\footnote{An extractor is said to be strong in a set of its sources if even conditioned on all the sources in this set the output cannot be distinguished from uniform (see formal definition in Section~\ref{sec:markov_chain_extractors}).} extractor is also an extractor in the model of \cite{chung2014multi}---a more precise comparison to previous work is given in Section~\ref{sec:intro.related}. Secondly, we consider it much more natural to put a requirement on the structure of the global state $\rho_{X_1X_2C}$, instead of describing permissible adversarial strategies that generate the side information $C$, as in~\cite{kasher2010two,chung2014multi}. Thirdly, Markov chains arise naturally in certain applications. For example, in realisations of quantum randomness amplification protocols one can sometimes assume that the devices on which the experiment is being preformed have a Markov chain structure (for further details see Section~\ref{sec:randomness_amplification}).

We also show that extractors in the Markov model can be used to extract randomness from a larger set of states. We prove that a bound on the \emph{smooth} min-entropy~\cite{renner2008security} suffices for randomness extraction. This can be seen as a robustness property of the model, since in many applications one can only bound the smooth min-entropy rather than the min-entropy itself. In addition, we prove that any CPTP map performed on the side information---which might delete information and thus destroy the Markov property---cannot decrease the security of an extractor, hence extractors in the Markov model are also extractors for such non-Markov states\footnote{This includes, in particular, states constructed according to the model of~\cite{chung2014multi}.}.

Our second contribution is to prove that \emph{all} extractors (weak and strong) remain secure in this model, both in the classical and quantum case, albeit with weaker parameters. In the classical case the proof is pretty trivial and standard (and the result is indeed not surprising). Nevertheless, as we could not find it anywhere else in the literature, we give it in this work for completeness and as comparison to the quantum case.  
More specifically, for classical side information we prove the following theorem:

\begin{thm}\label{thm:l-sources.mc}
	Any $(k_1,\dotsc,k_l,\varepsilon)$-[strong] $l$-source extractor is a  $\left(k_1 + \log \frac{1}{\varepsilon},\dotsc, k_l + \log\frac{1}{\varepsilon},(l+1)\varepsilon\right)$-[strong] classical-proof $l$-source extractor in the Markov model.
\end{thm}

The formal definitions of a (strong) $l$-source extractor and a (strong) classical-proof $l$-source extractor are given in Section~\ref{sec:classical-markov}. The important thing to note is that for the extractor to remain secure, the min-entropy of the sources needs to be just $\log\frac{1}{\varepsilon}$ higher, where $\varepsilon$ is the security parameter (or the error) of the extractor. This is exactly the same as in the case of seeded extractors~\cite{konig2008bounded} with classical side information. 

The main contribution of the current work is the quantum version of the theorem above:

\begin{thm}\label{thm:quantum_proof_extractors}
	Any $(k_1,\dotsc,k_l,\varepsilon)$-[strong] $l$-source extractor is a $\left(k_1 + \log \frac{1}{\varepsilon},\dotsc, k_l + \log\frac{1}{\varepsilon},\sqrt{(l+1)\varepsilon 2^{(m-2)}}\right)$-[strong] quantum-proof $l$-source extractor in the Markov model, where $m$ is the output length of the extractor.
\end{thm}
 
The formal definitions of the extractors are given in Section~\ref{sec:quantum_markov}. As in the classical case, the min-entropy of the sources needs to be just $\log\frac{1}{\varepsilon}$ higher. The error itself is $\sqrt{(l+1)\varepsilon 2^{(m-2)}}$ where $l$ is the number of sources and $m$ is the number of output bits of the considered extractor\footnote{This matches exactly the bound proven in~\cite[Corollary 27]{kasher2010two} for the restricted case of product side information, $l=2$, and $m=1$. We note that this is also an improvement over the constructions in the model of~\cite{chung2014multi}, for which the error in~\cite[Theorem 5.3]{chung2014multi} for $l=2$ is of the form $2^m\sqrt{\varepsilon}$, i.e., an order of $\sqrt{2^m}$ worse than ours.}.

Although the blow-up in the error of the extractor in Theorem~\ref{thm:quantum_proof_extractors} might seem relatively high, one must note that many classical multi-source extractors have an error $\eps = 2^{-mc}$ for some constant $c > 1$, hence in the quantum case the new error is $\eps' =\frac{\sqrt{l+1}}{2} 2^{-m\frac{c-1}{2}}$, i.e., both the classical and quantum errors are of the order $2^{-\Omega(m)}$. We provide several explicit constructions in Section~\ref{sec:specific_constructions}, where we show how to achieve similar parameters to the classical case, even if $\eps \gg 2^{-m}$, by composing the multi-source extractor with a quantum-proof seeded extractor. 

Apart from presenting the Markov model for extractors and proving the theorems above, we also contribute on the technical level. While previous works use the techniques of~\cite{konig2008bounded} for the one-bit output case and then extend it using a quantum XOR lemma~\cite{kasher2010two}, we use a completely different proof technique which is based on the recent work of~\cite{berta2014quantum,quantumbilinear}. The advantage of our technique is that it also applies to weak extractors, whereas the techniques of~\cite{chung2014multi} require the extractors to be strong in order to prove that they are secure.  We extend on our proof technique in Section~\ref{sec:intro.technique}.

\subsection{Related work}
\label{sec:intro.related}

As far as we are aware, the question of the security of multi-source extractors in the presence of side information was considered only in two works: \cite{kasher2010two} and \cite{chung2014multi}. Both works deal with quantum side information, and classical side information can of course be taken as a special case. We are not aware of any works dealing with the case of classical side information directly. 

\cite{kasher2010two} initiated the study of multi-source extractors in the presence of side information. They considered the case of two sources and quantum side information in product form. More specifically, given the two independent sources $X_1$ and $X_2$, the side information is given by a state $\rho_{C_1}\otimes\rho_{C_2}$ such that $H_{\text{min}}(X_i|C_1C_2) = H_{\text{min}}(X_i|C_i) \geq k_i$. In this way, the side information does not break the independence of the sources\footnote{\cite{kasher2010two} also considered another model for the adversary, called the bounded storage model, in which an assumption is made on the size of the adversary's storage capacity. In this work we consider only the more general case, in which we make no assumption about the adversary's power. For more details see~\cite{kasher2010two}.}. It was proven in~\cite{kasher2010two} that any \emph{one-bit output} two-source extractor remains secure in the presence of product side information. They further show that a specific construction of a multi-bit output two-source extractor, that of~\cite{dodis2004improved}, is also secure in the considered model, by reducing it to the one-bit case.

Recently, another, more general model for an adversary was considered in~\cite{chung2014multi}. For simplicity, we explain here the model for the case of two sources only; see~\cite{chung2014multi} for the general definition. In~\cite{chung2014multi} the side information of the adversary must be created in the following way: in the beginning the adversary can have any bipartite quantum state $\rho_{E_1E_2}$, independent of the sources. Then, to create her final side information $\rho_{C_1C_2}$, she can correlate her state with the sources by performing an independent ``leaking operation'' from each source to one of the subsystems. More specifically, they model the leaking operation as a map for $i\in\{1,2\}$, $\Phi_i : \mathrm{L}(X_i \otimes E_i) \rightarrow \mathrm{L} (X_i \otimes C_i)$. The resulting classical-quantum state $\rho_{X_1X_2C_1C_2}$ can be written as $\rho_{X_1X_2C_1C_2} = \Phi_1 \otimes \Phi_2 (\rho_{X_1X_2E_1E_2})$. For the relevant conditions on the min-entropy see~\cite{chung2014multi}.

It was then proven in~\cite{chung2014multi} that for multi-source extractors which are strong in all but one source, this complex adversarial leaking operation is in fact equivalent to providing the adversary with side information about only one source.  That is, when using an extractor which is strong in all but one of its sources, any adversary who is restricted to the model of~\cite{chung2014multi} is in fact no stronger than an adversary who has side information about just one source.  It is further shown that several known extractor constructions are still secure when the adversary holds quantum side information about one of the sources---with an increase in the error of the extractor.  The leaking model of~\cite{chung2014multi} can also be defined for weak extractors. However, the proof techniques of~\cite{chung2014multi} only work for strong extractors, since they rely on the equivalence to side information about one source. Thus, there are currently no known extractor constructions that directly satisfy the weak extractor model from~\cite{chung2014multi}, without relying on an underlying strong extractor.

Our work is a natural generalization of \cite{kasher2010two}, since independent sources are a subset of Markov sources. The model from \cite{chung2014multi} is different from ours in the sense that there exist states $\rho_{X_1X_2C}$ which are Markov chains but cannot be constructed by the leaking model from \cite{chung2014multi} and vice versa. However, as already proven in \cite{chung2014multi}, for a function to satisfy their strong extractor definition, it is sufficient for it to be secure in the presence of side information about one of the sources. Since side information about one source is a Markov chain, it follows that any strong extractor in the Markov model is also a strong extractor in the leaking model of~\cite{chung2014multi}---for completeness, we provide a proof of this in Section~\ref{sec:nonmarkov_sources}. It is currently unknown whether the same statement holds for weak extractors. Interestingly, the converse statement also holds: we (implicitly) prove in this work that any function that is an extractor for side information in product form is an extractor in the Markov model with slightly weaker parameters. Since the leaking model from \cite{chung2014multi} includes states in product from, an extractor from \cite{chung2014multi} is also an extractor in the Markov model with slightly weaker parameters.

\subsection{Proof outline and techniques}\label{sec:intro.technique}

The proof of the classical result, i.e., Theorem~\ref{thm:l-sources.mc}, is quite standard. The main part of this work is therefore devoted to the quantum case---the proof of Theorem~\ref{thm:quantum_proof_extractors}. The main idea is to not consider the most general measurement that could be performed to distinguish the output of the extractor from uniform, but instead consider a specific strategy, which consists in first measuring the quantum side information, then trying to distinguish the output from uniform given the resulting classical side information. We first prove that this specific strategy is not much worse than the optimal strategy. Then we show that this classical side information satisfies the requirements of a classical two-source extractor in the Markov model. Thus, security in the quantum case follows from security in the classical case.

More specifically, the proof can be decomposed in the following steps.
\begin{enumerate}
	\item We start by considering only \emph{product} side information in Section~\ref{sec:product}. We employ ideas from~\cite{berta2014quantum,quantumbilinear}, where the security definition of the extractor is rewritten using operators inequalities, to give a bound in Lemma~\ref{lem:CS_weak} on the distance from uniform of the extractor output.
	\item Next (in Lemmas~\ref{lem:simple_bound_weak} and~\ref{lem:distinguishing_weak}) we simplify the bound by noting that it can be seen as a \emph{specific} simple distinguishing strategy when trying to distinguish the output of the extractor from uniform using the side information. This specific strategy is one in which the product side information is measured independently of the output of the extractor (while a general distinguisher could use more complicated distinguishing strategies). Hence we obtain a reduction from quantum to classical side information.
	\item We put this together in Lemma~\ref{lem:classical_markov_to_product_quantum}, to show that $\emph{any}$ multi-source extractor is secure in the presence of product side information\footnote{This can be seen as an extension of the result of~\cite{kasher2010two} but the proof is different.}.
	\item Finally, in Section~\ref{sec:extension_to_markov_chain}, we  extend the result from the product model to the quantum Markov model by exploiting the structure of quantum Markov-chain states, and by this prove that Theorem~\ref{thm:quantum_proof_extractors} holds.
 \end{enumerate}

 \paragraph{Organisation of the paper.} The rest of paper is organised as follows. In Section~\ref{sec:prel} we give some necessary preliminaries. Section~\ref{sec:markov_chain_extractors} is devoted to the definitions of classical and quantum-proof multi-source extractors in the Markov model. The proof of our main theorem, Theorem~\ref{thm:quantum_proof_extractors}, is then given in Section~\ref{sec:security_proof}. In Section~\ref{sec:extending_sources} we show that multi-source extractors in the Markov model can be used to extract from some sources that do not directly satisfy the definition. In Section~\ref{sec:specific_constructions} we give the parameters of explicit constructions of quantum multi-source extractors, i.e., we apply our results to some specific constructions of multi-source extractors.  In Section~\ref{sec:randomness_amplification} we further motivate the Markov model in the context of quantum randomness amplification protocols. We conclude in Section~\ref{sec:open_questions} with some open questions.

\section{Preliminaries}\label{sec:prel}

We assume familiarity with standard notation in probability theory as well as with basic concepts in quantum information theory including density matrices, positive-operator valued measures (POVMs), and distance measures such as the trace distance. We refer to, e.g.,~\cite{nielsen2010quantum} for an introduction to quantum information.

Throughout the paper $X,Y$ and $Z$ denote classical random variables while $A,B$ and $C$ denote quantum systems. All logarithms are in base 2.  $[l]$ denotes the set $\{1,2,\dotsc,l\}$ and for $i\in[l]$ we denote $\bar{i} = 1,\dotsc,i-1,i+1,\dotsc,l$.

If a classical random variable $X$ takes the value $x$ with probability $p_x$ it can be written as the quantum state $\rho_X = \sum_x p_x \ketbra{x}{x}$, where $\{\ket{x}\}_x$ is an orthonormal basis. If the classical system $X$ is part of a composite system $XC$, any state of that composite system can be written as $\rho_{XC} =  \sum_x p_x \ketbra{x}{x}\otimes \rho_C^x$. If $C$ is quantum we say that the state $\rho_{XC}$ is a classical-quantum state, or a cq-state. Similarly, a state $\rho_{X_1X_2C}$ classical on $X_1,X_2$ and quantum on $C$ is called a ccq-state. For two independent random variables $X$ and $Y$ we often write $X \circ Y$ to denote the joint random variable with product distribution.

For a quantum state $\rho_A$ we denote by $H(A)$ the Von Neumann entropy of $\rho_A$, i.e., $H(A) = -\tr(\rho_A \log \rho_A)$. The  conditional mutual information is defined as
\[
	I(A:B|C)= H(AC) + H(BC) - H(C) - H(ABC) \;.
\]
In the case of classical systems, the Von Neumann entropy is reduced to the Shannon entropy. That is, for a random variable $X$, $H(X)= - \sum_x p_x \log p_x$, where $p_x$ is the probability of $X=x$. 

Given a cq-state $\rho_{XC}=  \sum_x p_x \ketbra{x}{x}\otimes \rho_C^x$ the conditional min-entropy is $H_{\text{min}}(X|C) = -\log p_{\text{guess}}(X|C)$, where $p_{\text{guess}}(X|C)$ is the maximum probability of guessing $X$ given the quantum system $C$. That is, 
\[
	p_{\text{guess}}(X|C) = \max_{\{E^x_C\}_x} \left( \sum_x p_x \tr (E^x_C\rho^x_C) \right) \;,
\]
where the maximum is taken over all POVMs $\{E^x_C\}_x$ on $C$. For an empty system $C$, the conditional min-entropy of $X$ given $C$ reduces to the usual $H_{\text{min}}(X)= -\log \max_x p_x$. 
Furthermore, if a quantum system $C$ is measured and the measurement outcome is registered in the classical system $Z$ then the min-entropy can only increase, namely, $H_{\text{min}}(X|Z) \geq H_{\text{min}}(X|C)$.

\section{Multi-source extractors}\label{sec:markov_chain_extractors}
\subsection{Multi-source extractors in the presence of classical side information}\label{sec:classical-markov}

Two-source extractors are defined as follows. The extension of the definition to the case of more than two sources is straightforward. 

\begin{definition}[Two-source extractor, \cite{raz05extractors}]\label{def:two-source}
	A function $\mathrm{Ext}:\{0,1\}^{n_1}\times\{0,1\}^{n_2}\rightarrow\{0,1\}^m$ is called a $(k_1,k_2,\varepsilon)$ two-source extractor if for any two independent sources $X_1,X_2$ with $H_{\text{min}}\left(X_1\right)\geq k_1$ and $H_{\text{min}}\left(X_2\right)\geq k_2$, we have 
	\[
		\frac{1}{2}\| \mathrm{Ext}\left( X_1,X_2 \right) - U_m \| \leq \varepsilon \;,
	\]
	where $U_m$ is a perfectly uniform random variable on $m$-bit strings. $\mathrm{Ext}$ is said to be \emph{strong in the $i$'th input} if
  	\[
		\frac{1}{2} \|  \mathrm{Ext}(X_1,X_2)X_i - U_m \circ X_i \| \leq \varepsilon \;. 
	\]
	 If $\mathrm{Ext}$ is not strong in any of its inputs it is said to be weak. 
\end{definition}

As explained in Section~\ref{sec:intro}, in the classical case one can also consider the security of the extractor in the presence of classical side information, denoted by $Z$, held by an adversary. That is, we would like the output of the extractor to be indistinguishable from uniform also \emph{given} some additional classical information.  

Since multi-source extractors cannot remain secure in the presence of an arbitrary classical side information (recall the examples presented in Section~\ref{sec:intro}), we require the sources to be independent conditioned on the side information. Formally:

\begin{definition}[Classical Markov model]\label{def:classical_markov_model}
	The random variables $X_1, X_2$ and $Z$ are said to form a Markov chain, denoted by $X_1 \leftrightarrow Z \leftrightarrow X_2$, if 
	\[
		I(X_1:X_2|Z) = 0 \;.
	\]
	For more than two sources $X_1,\dotsc,X_l$ and side information $Z$, we say that they are in the Markov model if
	\[
		\forall i \in [l], \quad I(X_i:X_{\bar{i}}|Z) = 0 \;.
	\]
\end{definition}

To see that $I(X_1:X_2|Z) = 0$ indeed captures the idea that conditioned on $Z$ the sources are independent, note that $I(X_1:X_2|Z) = 0$ if and only if $\mathrm{p}(x_1,x_2|z) = \mathrm{p}(x_1|z)\cdot\mathrm{p}(x_2|z)$ for all $x_1,x_2$ and $z$. 

We can now define classical-proof multi-source extractors in the following way. For simplicity, we give the definition in the case of two sources; the extension to more than two sources in the Markov model is straightforward. 

\begin{definition}[Classical-proof two-source extractor]\label{def:two-source.mc}
	A function $\mathrm{Ext}: \{0,1\}^{n_1} \times \{0,1\}^{n_2} \to \{0,1\}^m$ is a $(k_1,k_2,\varepsilon)$ classical-proof two-source extractor secure in the Markov model, if for all sources $X_1,X_2$, and classical side information $Z$, where $X_1 \leftrightarrow Z \leftrightarrow X_2$ form a Markov chain, and with min-entropy $H_{\text{min}}\left(X_1|Z\right)\geq k_1$ and $H_{\text{min}}\left(X_2|Z\right)\geq k_2$, we have
	\begin{equation}\label{eq:classical_ext_error_weak}
		\frac{1}{2}\| \mathrm{Ext}\left( X_1,X_2 \right)Z - U_m\circ Z \| \leq \varepsilon \;,
	\end{equation}
	where $U_m$ is a perfectly uniform random variable on $m$-bit strings. $\mathrm{Ext}$ is said to be \emph{strong in the $i$'th input} if
  	\begin{equation}\label{eq:classical_ext_error_strong}
	 	\frac{1}{2} \|  \mathrm{Ext}(X_1,X_2)X_iZ - U_m \circ X_i Z \| \leq \varepsilon \;. 
	\end{equation}
\end{definition}

Indeed, if one requires that the sources and the side information $Z$ fulfil Definition~\ref{def:classical_markov_model} then all multi-source extractors remain secure also in the presence of the side information $Z$. This is proven in the following lemma for two sources.

\begin{lemma}\label{lem:two-source.mc}
	Any $(k_1,k_2,\varepsilon)$-[strong] two-source extractor is a  $(k_1 + \log \frac{1}{\varepsilon}, k_2 + \log\frac{1}{\varepsilon},3\varepsilon)$-[strong] classical-proof two-source extractor in the Markov model.
\end{lemma}

\begin{proof}
	Let $X_1 \leftrightarrow Z \leftrightarrow X_2$ be such that $H_{\text{min}}\left(X_1|Z\right)\geq k_1 + \log \frac{1}{\varepsilon}$ and $H_{\text{min}}\left(X_2|Z\right) \geq k_2 + \log \frac{1}{\varepsilon}$. For any two classical systems $X$ and $Z$, we have 
	\[
		2^{-H_{\text{min}}\left(X|Z\right)} = \mathop{\mathbb{E}}_{z \leftarrow Z} \left[2^{-H_{\text{min}}\left(X|Z=z\right)}\right] \;,
	\] 
	so by Markov's inequality, 
	\[
		\Pr_{z \leftarrow Z} \left[ H_{\text{min}}\left(X|Z=z\right)\leq H_{\text{min}}\left(X|Z\right) - \log 1/\varepsilon \right] \leq \varepsilon \;.
	\]
	Applying this to both $X_1$ and $X_2$, we have that with probability at least $1-2\varepsilon$ (over $Z$), $H_{\text{min}}\left(X_1|Z=z\right) \geq k_1$ and $H_{\text{min}}\left(X_2|Z=z\right) \geq k_2$. Due to the Markov-chain condition, the distributions $X_1|_{Z=z}$ and $X_2|_{Z=z}$ are independent. Hence for any $(k_1,k_2,\varepsilon)$ two-source extractor $\mathrm{Ext}$,
	\[
    		\frac{1}{2} \| \mathrm{Ext}(X_1,X_2)Z - U_m \circ Z\| = \frac{1}{2} \sum_{z} P_Z(z) \|\mathrm{Ext}(X_1|_{Z=z},X_2|_{Z=z}) - U_m\|  \leq 3 	\varepsilon \;.
	\]
 	For a strong extractor the proof is identical.
\end{proof}

By following the same steps as the proof of Lemma \ref{lem:two-source.mc} for the case of $l$ sources we get Theorem~\ref{thm:l-sources.mc}.

\subsection{Multi-source extractors in the presence of quantum side information}\label{sec:quantum_markov}

We now consider multi-source extractors in the presence of quantum side information, i.e., in the following $C$ denotes a quantum system. Similarly to Section \ref{sec:classical-markov} we restrict the sources and the quantum side information to the quantum Markov model. Formally, 

\begin{definition}[Quantum Markov model]\label{def:quantum_markov_mutual_info}
  	A ccq-state \(\rho_{X_1X_2C}\) is said to form a Markov chain\footnote{The same definition is also used in the more general case where also the $X_i$'s are quantum. For our purpose the case of classical sources and quantum side information is sufficient.}, denoted by $X_1 \leftrightarrow C \leftrightarrow X_2$, if 
	\[
		I(X_1:X_2|C) = 0 \;.
	\]
	For more than two sources $X_1,\dotsc,X_l$ and $C$ we say that they are in the Markov model if
	\[
		\forall i \in [l], \quad I(X_i:X_{\bar{i}}|C) = 0 \;.
	\]
\end{definition}

The following is then the natural analog of Definition~\ref{def:two-source.mc} to the quantum setting. The extension to the case of more than two sources is straightforward.
\begin{definition}[Quantum-proof two-source extractor]\label{def:quantum-two-source.mc}
	A function $\mathrm{Ext}: \{0,1\}^{n_1} \times \{0,1\}^{n_2} \to \{0,1\}^m$ is a $(k_1,k_2,\varepsilon)$ quantum-proof two-source extractor in the Markov model, if for all sources $X_1,X_2$, and quantum side information $C$, where $X_1 \leftrightarrow C \leftrightarrow X_2$ form a Markov chain, and with min-entropy $H_{\text{min}}\left(X_1|C\right)\geq k_1$ and $H_{\text{min}}\left(X_2|C\right)\geq k_2$, we have
	\begin{equation}\label{eq:ext_error_weak}
		\frac{1}{2}\| \rho_{\mathrm{Ext}(X_1,X_2)C} - \rho_{U_m} \otimes \rho_{C} \| \leq \varepsilon \;,
	\end{equation}
	where $\rho_{\mathrm{Ext}(X_1,X_2)C} = \Ext \otimes \mathbb{1}_C \rho_{X_1X_2C}$ and $\rho_{U_m}$ is the fully mixed state on a system of dimension $2^m$. $\mathrm{Ext}$ is said to be \emph{strong in the $i$'th input} if
  	\begin{equation}\label{eq:ext_error_strong}
	 	\frac{1}{2}\| \rho_{\mathrm{Ext}(X_1,X_2)X_iC} - \rho_{U_m} \otimes \rho_{X_iC} \| \leq \varepsilon \;. 
	\end{equation}
\end{definition}
If $C$ above is classical then Definition~\ref{def:quantum-two-source.mc} is reduced to Definition~\ref{def:two-source.mc}.
 
The interesting question is therefore whether there exist quantum-proof multi-source extractors. The main contribution of this work is to show that \emph{any} multi-source extractor is also a quantum-proof multi-source extractor in the Markov model with a bit weaker parameters. The formal statement is given in Theorem~\ref{thm:quantum_proof_extractors} above and proven in the following section. 

\section{Security of multi-source extractors in the quantum Markov model}\label{sec:security_proof}

For simplicity, in this section we prove that two-source extractors are secure even when considering quantum side information in the form of a Markov chain. The extension to any number of sources, i.e., the proof of Theorem~\ref{thm:quantum_proof_extractors}, follows by trivially repeating the same steps for more than two sources and using our definition of the Markov model (Definition~\ref{def:quantum_markov_mutual_info}). More specifically, the goal of this section is to prove the following:
\begin{lemma}\label{lem:quantum_markov_two_source}
	Any $(k_1,k_2,\varepsilon)$-[strong] two-source extractor is a $\left(k_1 + \log \frac{1}{\varepsilon}, k_2 + \log\frac{1}{\varepsilon},\sqrt{3\varepsilon \cdot 2^{(m-2)}}\right)$-[strong] quantum-proof two-source extractor in the Markov model, where $m$ is the output length of the extractor.
\end{lemma}

To prove this, we first show in Section~\ref{sec:product} that all extractors are still secure in the case of side information in product from. Then in Section~\ref{sec:extension_to_markov_chain} we generalise this result to any side information in the Markov model.

\subsection{Product quantum side information}\label{sec:product}

We start by showing that \emph{any} two-source extractor, as in Definition~\ref{def:two-source}, is secure against product quantum side information. The product extractor as defined below is a special case of the extractor in Definition~\ref{def:quantum-two-source.mc}:

\begin{definition}[Quantum-proof product two-source extractor, \cite{kasher2010two}]\label{def:product_extractor}
	A function $\mathrm{Ext}: \{0,1\}^{n_1} \times \{0,1\}^{n_2} \to \{0,1\}^m$ is a $(k_1,k_2,\varepsilon)$ quantum-proof product two-source extractor, if for all sources $X_1,X_2$, and quantum side information $C$, where $\rho_{X_1X_1C} = \rho_{X_1C_1}\otimes\rho_{X_2C_2}$, and with min-entropy $H_{\text{min}}\left(X_1|C_1\right)\geq k_1$ and $H_{\text{min}}\left(X_2|C_2\right)\geq k_2$, we have
	\begin{equation}\label{eq:ext_error_weak.product}
		\frac{1}{2}\| \rho_{\mathrm{Ext}(X_1,X_2)C} - \rho_{U_m} \otimes \rho_{C} \| \leq \varepsilon \;,
	\end{equation}
	where $\rho_{\mathrm{Ext}(X_1,X_2)C} = \Ext \otimes \mathbb{1}_{C} \rho_{X_1X_2C}$ and $\rho_{U_m}$ is the fully mixed state on a system of dimension $2^m$. $\mathrm{Ext}$ is said to be \emph{strong in the $i$'th input} if
  	\begin{equation}\label{eq:ext_error_strong.product}
	 	\frac{1}{2}\| \rho_{\mathrm{Ext}(X_1,X_2)X_iC} - \rho_{U_m} \otimes \rho_{X_iC_i} \otimes \rho_{C_{\bar{i}}} \| \leq \varepsilon \;. 
	\end{equation}
\end{definition}

In the following we show that any two-source extractor remains secure in the product model, i.e., if the quantum state of the sources and the side information is of the form $\rho_{X_1X_2C} = \rho_{X_1C_1}\otimes\rho_{X_2C_2}$ (see Corollary~\ref{cor:two_source_to_product_quantum} below for the formal statement).
This can be seen as an extension of the results of~\cite{kasher2010two}, where only two-source extractors with one-bit output (i.e., $m=1$ in our notation) and the extractor of~\cite{dodis2004improved} were shown to be secure against product quantum side information.

The first step of the proof uses the fact that any ccq-state $\rho_{X_1X_2C}$ can be obtained by performing local measurements on a pure state $\rho_{ABC}$. We formalise this in the following lemma. The proof of the lemma is trivial and given in Appendix \ref{sec:proofs_security_two}.

\begin{lemma}\label{lem:classical_to_quantum_state}
	Let $\rho_{X_1X_2C} = \sum_{x_1,x_2} \proj{x_1} \otimes \proj{x_2} \rho_C(x_1,x_2)$. Then there exists a pure state $\rho_{ABC}$ and POVMs $\{F_{x_1}\},\{G_{x_2}\}$ such that
	\begin{equation}\label{eq:defomegacqstate}
 		\rho_C(x_1,x_2) = \tr_{AB}\left[ F_{x_1}^{\frac{1}{2}} \otimes G_{x_2}^{\frac{1}{2}} \otimes \mathbb{1}_{C}\rho_{ABC} F_{x_1}^{\frac{1}{2}} \otimes G_{x_2}^{\frac{1}{2}}  \otimes \mathbb{1}_{C}\right] \;.
	\end{equation}
\end{lemma}

The following three lemmas are proven for the case of weak extractors. The lemmas and proofs for the strong case are very similar and therefore given in Appendix~\ref{sec:strong_proofs}. 
We start with the next lemma where the Cauchy-Schwarz inequality is used, as in~\cite{berta2014quantum,quantumbilinear}.
\begin{lemma}\label{lem:CS_weak}
Let $\rho_{X_1X_2C}$ be any ccq-state, and let $\rho_{ABC}$ and $\{F_{x_1}\},\{G_{x_2}\}$ satisfy Equation~\eqref{eq:defomegacqstate}. Then there exists an alternative purification of $\rho_{AB}$, namely \(\Psi_{ABC_1C_2}\), and two POVMs \(\{H_{z_1}\}\), \(\{K_{z_2}\}\) acting on \(C_1\) and \(C_2\), such that
  	\begin{multline*}
   		\frac{1}{M}\| \rho_{\mathrm{Ext}(X_1,X_2)C} - \rho_{U_m} \otimes \rho_{C}\|^2 \leq\\
   		 \sum_{\substack{x_1,x_2,\\z_1,z_2,y}} \left[ \delta_{\mathrm{Ext}\left(x_1,x_2\right)=y} - \frac{1}{M}\right]  				\left[ \delta_{\mathrm{Ext}\left(z_1,z_2\right)=y} - \frac{1}{M}\right]  \; \tr\left[\Psi_{ABC_1C_2} F_{x_1} \otimes G_{x_2} \otimes H_{z_1} \otimes K_{z_2} \right]\,,
 	 \end{multline*}
         where $M=2^m$ and $m$ is the output length of the extractor. Moreover, if the state \(\rho_{AB}\) is of tensor product form, the purification \(\Psi_{ABC_1C_2}\) also factorises into a tensor product between \(AC_1\) and \(BC_2\).  
\end{lemma}

\begin{proof}
  First, recall that for a hermitian matrix \(R\) we have \(\norm{R} = \max \{\tr[RS] \,:\,-\mathbb{1}\leq S \leq \mathbb{1}\}\). Applying this to the matrix whose norm specifies the error of the extractor, we find
  \[
    \| \rho_{\mathrm{Ext}(X_1,X_2)C} - \rho_{U_m}\otimes \rho_{C}\| = \max_{-\mathbb{1}\leq S \leq \mathbb{1}}\tr\left[ \left(\rho_{\mathrm{Ext}(X_1,X_2)C} - \rho_{U_m}\otimes \rho_{C} \right) S\right] \,.
  \]
  Since $\rho_{\mathrm{Ext}(X_1,X_2)C}$ and $\rho_{U_m}\otimes \rho_{C}$ are block diagonal with respect to the outcome variable of the extractor \(y\), \(S\) can be assumed to be block diagonal as well. Using this and inserting the expression for \(\rho_{X_1X_2C}\) in Equation \eqref{eq:defomegacqstate} we arrive at
  \[
    \| \rho_{\mathrm{Ext}(X_1,X_2)C} - \rho_{U_m}\otimes \rho_{C}\| = \max_{-\mathbb{1} \leq S_y \leq \mathbb{1} }\,\sum_{y} \,\tr\left[\rho_{ABC} \Delta_{y} \otimes S_y\right]\,,
  \]
  where we used the abbreviation
  \[
    \Delta_{y} = \sum_{x_1,x_2} \left[ \delta_{\mathrm{Ext}\left(x_1,x_2\right)=y} - \frac{1}{M}\right] F_{x_1}\otimes G_{x_2} \,.
  \]
  We now choose a special purification of \(\rho_{AB}\), namely we consider the \emph{pretty good purification}~\cite{winter2004extrinsic}
	\[
    		\ket{\psi}_{ABA'B'} = \rho_{AB}^\half \otimes \mathbb{1}_{A^\prime B^\prime} \, \ket{\Phi_{AA^\prime}}\ket{\Phi_{BB^\prime}}\,,
	\]
  where \(\ket{\Phi}_{AA^\prime} = \sum_{a} \ket{aa}\) denotes the unnormalised maximally entangled state. Since both $\ket{\psi}_{ABA'B'} $ and $\rho_{ABC}$ are purifications of $\rho_{AB}$ there exists an isometry \(V: A^\prime B^\prime \to C\) such that $V\kettbra{\psi} V^* = \rho_{ABC}$ and hence
  \[
    \| \rho_{\mathrm{Ext}(X_1,X_2)C} - \rho_{U_m}\otimes \rho_{C}\| \leq \max_{-\mathbb{1} \leq S_y \leq \mathbb{1} }\,\sum_{y}   \,\tr\left[\kettbra{\psi} \Delta_{y} \otimes S_y\right]\,,
  \]
  since \(V^* S_y V\) is bounded in norm by one and hermitian. Inserting the identity \(\mathbb{1} \otimes X_{A^\prime} \ket{\Phi_{AA^\prime}} = X_{A^\prime}^T \otimes \mathbb{1} \ket{\Phi_{AA^\prime}}\) for any matrix \(X\) (where \((.)^T\) denotes the transpose in the basis of the maximally entangled state), we find
  \begin{align}\label{eq:bringestimateinscpform}
    \tr\left[\kettbra{\psi} \Delta_{y} \otimes S_y\right] = \tr\left[\rho_{AB}^\half \Delta_{y} \rho_{AB}^\half (S_y)^T\right]\,.
  \end{align}
  The crucial observation is now that the sesquilinear form $(R_y) \times (S_y) \mapsto \sum_y \tr\left[\rho_{AB}^\half R^*_y \rho_{AB}^\half S_y\right]$ on block-diagonal matrices is positive semi-definite and hence fulfils the Cauchy-Schwarz inequality. Applying this gives
  \[
    \| \rho_{\mathrm{Ext}(X_1,X_2)C} - \rho_{U_m}\otimes \rho_{C}\|^2 \leq \left(\sum_y \tr\left[\rho_{AB}^\half\Delta_{y} \rho_{AB}^\half \Delta_y \right] \right) \left( \sum_y \tr\left[\rho_{AB}^\half S_{y} \rho_{AB}^\half S_y\right]\right)\,.
  \]
  Since we have that the norm of \(S_y\) is bounded by one, the terms in the second sum satisfies
  \[
    \tr\left[\rho_{AB}^\half S_{y} \rho_{AB}^\half S_y\right] \leq \tr\left[\rho_{AB}^\half S_{y} \rho_{AB}^\half\right] \leq \tr\left[\rho_{AB}\right] = 1\,,
  \]
  and we arrive at
  \[
    \| \rho_{\mathrm{Ext}(X_1,X_2)C} - \rho_{U_m}\otimes \rho_{C}\|^2 \leq M\,\sum_y \tr\left[\rho_{AB}^\half\Delta_{y} \rho_{AB}^\half \Delta_y \right]\,.
  \]
  Inserting the definition of \(\Delta_y\) and reversing the identity leading to Equation~\eqref{eq:bringestimateinscpform} proves the assertion with \(C_1 = A^\prime\), \(C_2 = B^\prime\), \(\Psi_{ABC_1C_2} = \kettbra{\psi}\) and \(H_{z_1} = F_{z_1}^T\), \(K_{z_2} = G_{z_2}^T\).
\end{proof}

The upper bound of the preceding lemma can be further simplified (the proof is given in Appendix \ref{sec:proofs_security_two}):
\begin{lemma}\label{lem:simple_bound_weak}
	For any $\Psi_{ABC_1C_2}$ and positive operators $\{F_{x_1}\},\{G_{x_2}\},\{H_{z_1}\},\{K_{z_2}\}$ which sum up to the identity,
	\begin{align}
		\sum_{\substack{x_1,x_2,\\z_1,z_2,y}} 
		\left[ \delta_{\mathrm{Ext}\left(x_1,x_2\right)=y} - \frac{1}{M}\right]  
		\left[ \delta_{\mathrm{Ext}\left(z_1,z_2\right)=y} - \frac{1}{M}\right]  \; 
		\tr\left[ \Psi_{ABC_1C_2} F_{x_1}\otimes G_{x_2} \otimes H_{z_1}\otimes K_{z_2} \right] \nonumber \\
		= \sum_{\substack{x_1,x_2,z_1,z_2 | \\ \mathrm{Ext}\left(x_1,x_2\right)=\mathrm{Ext}\left(z_1,z_2\right)}} \tr\left[ \Psi_{ABC_1C_2} F_{x_1}\otimes G_{x_2} \otimes H_{z_1}\otimes K_{z_2} \right] - \frac{1}{M} \label{eq:specific_attack_weak}
	\end{align}
\end{lemma}

The quantity in Equation~\eqref{eq:specific_attack_weak} can be seen as a simple distinguishing strategy of a distinguisher trying to distinguish the output of the extractor from uniform given classical side information. We can therefore relate it to the error of the extractor in the case of classical side information, i.e., to Equation~\eqref{eq:classical_ext_error_weak}. This is shown in the following lemma. 

\begin{lemma}\label{lem:distinguishing_weak}
	For $i\in\{1,2\}$ let $Z_i$ denote the classical side information about the source $X_i$ such that $\mathrm{p}(x_1,x_2,z_1,z_2) = \tr\left[\Psi_{ABC_1C_2} F_{x_1}\otimes G_{x_2} \otimes H_{z_1}\otimes K_{z_2} \right]$. Then
	\[
		\sum_{\substack{x_1,x_2,z_1,z_2 | \\ \mathrm{Ext}\left(x_1,x_2\right)=\mathrm{Ext}\left(z_1,z_2\right)}} \mathrm{p}(x_1,x_2,z_1,z_2) - \frac{1}{M} \leq \frac{1}{2}\| \mathrm{Ext}\left( X_1,X_2 \right)Z_1Z_2 - U_m\circ Z_1Z_2 \| \;.
	\]
\end{lemma}
\begin{proof}
	Define the following random variables over $\{0,1\}^m\times \{0,1\}^{n_1} \times \{0,1\}^{n_2}$: 
	\[
		\mathrm{R} = \mathrm{Ext}(X_1,X_2)Z_1Z_2 \quad ; \quad \mathrm{Q} = U_m \circ Z_1Z_2 \;.
	\]
	Let $\mathcal{A^\star} = \left\{ (a_1,a_2,a_3) \big| a_1 = \mathrm{Ext}\left(a_2,a_3\right) \right\} \subseteq \{0,1\}^m\times \{0,1\}^{n_1} \times \{0,1\}^{n_2}$. Then, the probabilities that $\mathrm{R}$ and $\mathrm{Q}$ assign to the event $\mathcal{A^\star}$ are 
	\[
		\mathrm{R}(\mathcal{A^\star}) = \sum_{\substack{x_1,x_2,z_1,z_2 | \\ \mathrm{Ext}\left(x_1,x_2\right)=\mathrm{Ext}\left(z_1,z_2\right)}} \mathrm{p}(x_1,x_2,z_1,z_2) \quad ; \quad \mathrm{Q}(\mathcal{A^\star}) = \frac{1}{M}
	\]
	Using the definition of the variational distance we therefore have 
	\begin{align*}
		\frac{1}{2} \| \mathrm{Ext}(X_1,X_2)Z_1Z_2 - U_m \circ Z_1Z_2\|  &= \sup_{\mathcal{A}} \| \mathrm{R}(\mathcal{A}) - \mathrm{Q}(\mathcal{A}) \| \\
		& \geq  \mathrm{R}(\mathcal{A^\star}) - \mathrm{Q}(\mathcal{A^\star}) \\
		& = \sum_{\substack{x_1,x_2,z_1,z_2 | \\ \mathrm{Ext}\left(x_1,x_2\right)=\mathrm{Ext}\left(z_1,z_2\right)}} \mathrm{p}(x_1,x_2,z_1,z_2) - \frac{1}{M} \;. \qedhere
	\end{align*}
\end{proof}

Finally, we combine the lemmas together to show that any weak classical-proof two-source extractor in the Markov model is secure against product quantum side information as well.

\begin{lemma}\label{lem:classical_markov_to_product_quantum}
	Any $(k_1,k_2,\varepsilon)$ classical-proof two-source extractor in the Markov model is a $\left(k_1, k_2,\sqrt{\varepsilon \cdot 2^{(m-2)}}\right)$ quantum-proof product two-source extractor, where $m$ is the output length of the extractor.
\end{lemma}

\begin{proof}
	For any state of two classical sources and product side information $\rho_{X_1X_2C}=\rho_{X_1C_1}\otimes\rho_{X_2C_2}$ with $\Hmin{X_1|C} \geq k_1$ and $\Hmin{X_2|C} \geq k_2$,  let $\rho_{ABC}$ and $\{F_{x_1}\},\{G_{x_2}\}$ be the state and measurements satisfying Equation \eqref{eq:defomegacqstate}.

	We can now apply Lemmas \ref{lem:CS_weak}, \ref{lem:simple_bound_weak}, and \ref{lem:distinguishing_weak} to get the bound
	\begin{equation}\label{eq:combined_bound_weak}
		\| \rho_{\mathrm{Ext}(X_1,X_2)C} - \rho_{U_m} \otimes \rho_{C}\| \leq \sqrt{ \frac{M}{2} \| \mathrm{Ext}\left( X_1,X_2 \right)Z_1Z_2 - U_m\circ Z_1Z_2 \|} \;,
	\end{equation}
	where $Z_1,Z_2$ are defined via $\mathrm{p}(x_1,x_2,z_1,z_2) = \tr\left[\Psi_{ABC_1C_2}  F_{x_1}\otimes G_{x_2} \otimes H_{z_1}\otimes K_{z_2} \right]$, for $\Psi_{ABC_1C_2}$ which is constructed in the proof of Lemma \ref{lem:CS_weak}.
	
	As $\Psi_{ABC_1C_2} = \Psi_{AC_1} \otimes \Psi_{BC_2}$ and the measurements are all in tensor product we have $\mathrm{p}(x_1,x_2,z_1,z_2)=\mathrm{p}(x_1,z_1)\cdot\mathrm{p}(x_2,z_2)$, which implies: 
	\begin{enumerate}
		\item The sources and the classical side information form a Markov chain $X_1 \leftrightarrow Z_1Z_2 \leftrightarrow X_2$.
		\item $H_{\text{min}}\left(X_i|Z_1Z_2\right) = H_{\text{min}}\left(X_i|Z_i\right) \geq H_{\text{min}}\left(X_i|C_i\right)$ for $i\in\{1,2\}$. 
	\end{enumerate}
	
	Hence, if $H_{\text{min}}\left(X_i|C_i\right)\geq k_i$ then by the definition of a classical-proof two-source extractor,
	\begin{equation}\label{eq:class_extractor_promise_weak}
		\frac{1}{2} \| \mathrm{Ext}\left( X_1,X_2 \right)Z_1Z_2 - U_m\circ Z_1Z_2 \| \leq \varepsilon \;. 
	\end{equation}
	
	Combining Equations \eqref{eq:combined_bound_weak} and \eqref{eq:class_extractor_promise_weak} we get 
	\[
		\frac{1}{2} \| \rho_{\mathrm{Ext}(X_1,X_2)C} - \rho_{U_m} \otimes \rho_{C}\| \leq \frac{1}{2} \sqrt{M\varepsilon} = \sqrt{\varepsilon 2^{(m-2)}} \;. \qedhere
	\]
\end{proof}

By combining Lemma \ref{lem:two-source.mc} together with Lemma \ref{lem:classical_markov_to_product_quantum} (Lemma~\ref{lem:classical_markov_to_product_quantum_strong} in Appendix~\ref{sec:strong_proofs}) for the weak (strong) case we get that any weak (strong) two-source extractor is also secure against product quantum side information. The bound given in Corollary~\ref{cor:two_source_to_product_quantum} matches exactly the bound given in~\cite{kasher2010two} for the special case of $m=1$ (see~\cite[Corollary 27]{kasher2010two}).

\begin{cor}\label{cor:two_source_to_product_quantum}
	Any $(k_1,k_2,\varepsilon)$-[strong] two-source extractor is a $\left(k_1 + \log \frac{1}{\varepsilon}, k_2 + \log\frac{1}{\varepsilon},\sqrt{3\varepsilon \cdot 2^{(m-2)}}\right)-[strong]$ quantum-proof product two-source extractor, where $m$ is the output length of the extractor.
\end{cor}

\subsection{Extending to the Markov model}\label{sec:extension_to_markov_chain}

We now extend the result of Section \ref{sec:product} to the case of the more general Markov model. To do so, we first recall that by the result of~\cite{hayden2004structure}, Markov states (according to Definition \ref{def:quantum_markov_mutual_info}) can also be written in the form
\begin{equation} \label{eq:quantummarkov}
	\rho_{A_1A_2C} = \bigoplus_t \mathrm{p}(t) \rho^t_{A_1C^t_1} \otimes \rho^t_{A_2C^t_2}\,,
\end{equation}
where the index \(t\) runs over a finite alphabet \(T\), \(\mathrm{p}(t)\) is a probability distribution on that alphabet, \(\mathcal{H}_C=\bigoplus_t\mathcal{H}_{C_1^t} \otimes \mathcal{H}_{C_2^t}\) is the Hilbert space of $C$, and \(\rho^t_{A_iC^t_i}\) denote states on \(A_iC^t_i\), \(i\in\{1,2\}\).

\begin{proof}[Proof of Lemma~\ref{lem:quantum_markov_two_source}]

	Let $\rho_{X_1X_2C}$ be a Markov state (as in Definition \ref{def:quantum_markov_mutual_info}) such that $H_{\text{min}}\left(X_i|C\right) \geq k_i+\log \frac{1}{\eps}$. We first deal with the case of weak extractors. Using the decomposition from Equation \eqref{eq:quantummarkov} we can reduce the problem to the product case by writing
	\[
		\| \rho_{\mathrm{Ext}(X_1,X_2)C} - \rho_{U_m} \otimes \rho_{C}\| = \sum_t \mathrm{p}(t)  \| \mathrm{Ext} \otimes \mathbb{1}_C  \left(\rho^t_{X_1C^t_1} \otimes \rho^t_{X_2C^t_2}\right) - \rho_{U_m} \otimes \rho^t_{C^t_1} \otimes \rho^t_{C^t_2}\| \,.
	\]
	From Equation \eqref{eq:combined_bound_weak} we thus have
	\begin{align*}
		\| \rho_{\mathrm{Ext}(X_1,X_2)C} - \rho_{U_m} \otimes \rho_{C}\| & \leq \sum_t \mathrm{p}(t)  \sqrt{\frac{M}{2}\| \mathrm{Ext}(X_1,X_2)Z_1Z_2|T=t - U_m \circ Z_1Z_2|T=t \|} \\
		& \leq \sqrt{\frac{M}{2}\| \mathrm{Ext}(X_1,X_2)Z_1Z_2T - U_m \circ Z_1Z_2T \|}\,,
	\end{align*} 
	where in the last line we used Jensen's inequality and $Z_1,Z_2$ are defined via 
	\[
		\mathrm{p}(x_1,x_2,z_1,z_2|t) = \tr\left[\rho^t_{ABC} F^t_{x_1}\otimes G^t_{x_2} \otimes H^t_{z_1}\otimes K^t_{z_2} \right] \;.  
	\]
	That is, $Z_1$ and $Z_2$ are derived from $C$ in the following way: from Equation~\eqref{eq:quantummarkov} the states $\{\rho^t_{C_1^t} \otimes \rho^t_{C_2^t}\}_t$ are orthogonal, hence there exists an isometry $C \to CT$ which maps $\sum_t \mathrm{p}(t) \rho^t_{C_1^t} \otimes \rho^t_{C_2^t}$ to $\sum_t \mathrm{p}(t) \rho^t_{C_1^t} \otimes \rho^t_{C_2^t} \otimes \ketbra{t}{t}$. The state $\rho^t_{C_1^t} \otimes \rho^t_{C_2^t}$ is then measured in the same way as in Lemma~\ref{lem:classical_markov_to_product_quantum} for the product case to get the side information $Z_1Z_2|T$.
	Hence, the structure $X_1 \leftrightarrow Z_1Z_2T \leftrightarrow X_2$ is conserved. Furthermore, we also have $H_{\text{min}}\left(X_i|Z_1Z_2T\right) \geq k_i+\log \frac{1}{\eps}$. Using these two conditions, the problem has been reduced to one with classical side information in the Markov model. Using the fact that $\mathrm{Ext}$ is a $(k_1,k_2,\eps)$ two-source extractor and applying Lemma~\ref{lem:two-source.mc} we conclude the proof for weak extractors.
	
	Similarly, for strong extractors, from Equation \eqref{eq:combined_bound_strong} we have
	\begin{align*}
		\| \rho_{\mathrm{Ext}(X_1,X_2)X_1C} - \rho_{U_m} \otimes \rho_{X_1C}\| & \leq \sum_t \mathrm{p}(t)  \sqrt{\frac{M}{2}\| \mathrm{Ext}(X_1,X_2)X_1Z_2|T=t - U_m \circ X_1Z_2|T=t \|} \\
		& \leq \sqrt{\frac{M}{2}\| \mathrm{Ext}(X_1,X_2)X_1Z_2T - U_m \circ X_1Z_2T \|} \;.
	\end{align*} 
	Again, we can see this as a measurement made on $C$ such that the value of $T$ is measured and then a further measurements of $C_2^t$ is done in the same way as for the product case to get the side information about $X_2$ (while there is no additional side information about $X_1$). Hence, as in the weak case, $X_1 \leftrightarrow Z_2T \leftrightarrow X_2$ and $H_{\text{min}}\left(X_i|Z_2T\right) \geq k_i+\log \frac{1}{\eps}$, so the problem has been reduced to the classical case. 
\end{proof}

In the case of $l$ sources, a state $\rho_{X_{[l]}C}$ that satisfies the Markov model (Definition \ref{def:quantum_markov_mutual_info}) can be written as
\begin{equation} \label{eq:quantummarkov.multi} \rho_{X_{[l]}C} = \bigoplus_t \mathrm{p}(t) \rho^t_{X_1C_1^t} \otimes \dotsb \otimes \rho^t_{X_lC_l^t} \;.\end{equation} We provide a proof of this in Appendix~\ref{sec:proofs_security_two} as Lemma~\ref{lem:multi.markov}. It follows from Equation~\eqref{eq:quantummarkov.multi} that Lemma~\ref{lem:quantum_markov_two_source} can be easily generalised to $l$ sources.

\section{Extending the set of extractable sources}\label{sec:extending_sources}

Although the definition of a quantum-proof two-source extractors (Definition~\ref{def:quantum-two-source.mc}) requires the source $\rho_{X_1EX_2}$ to be a Markov chain with a bound on the min-entropy, a function proven be such an extractor can also be used to extract randomness from a larger set of sources, e.g., if the adversary were to destroy her side information $E$, this would not hinder extraction, yet it could destroy the Markov chain property of the source. In this section we consider two extensions of the multi-source extractor definition for which all multi-source extractors in the Markov model can be used. In Section~\ref{sec:smooth_min-entropy} we show that it is not necessary to have a bound on the min-entropy, it is sufficient to bound the smooth min-entropy of the sources $X_1$ and $X_2$. Then in Section~\ref{sec:nonmarkov_sources} we show that one can also extract from any source obtained by deleting information from a Markov source, even though the resulting state might not be a Markov chain any longer. The multi-source extractor model for strong extractors from \cite{chung2014multi} falls in this category.

\subsection{Smooth min-entropy}\label{sec:smooth_min-entropy}

It is standard for the extractor definitions to require a bound on the min-entropy of the source conditioned on the side information, i.e., $H_{\text{min}}(X_i|C) \geq k_i$. In practical situations, however, one often only has a bound on the \emph{smooth} min-entropy---this is defined by maximising the min-entropy over all states $\delta$-close, see Equation~\eqref{eq:smooth_min-entropy} below. For example, in quantum key distribution a bound on the smooth min-entropy is obtained by sampling the noise on the quantum channel~\cite{tomamichel2012QKD}. In this section we prove that any quantum-proof two-source extractor can be used in a context where only a bound on the smooth min-entropy is known.

The smooth conditional min-entropy with smoothness parameter $\delta$ of a state $\rho_{XC}$ is defined as follows.
\begin{equation}\label{eq:smooth_min-entropy}
H_{\text{min}}^\delta(X|C)_\rho = \max_{\sigma \in \mathcal{B}^\delta(\rho)} H_{\text{min}}(X|C)_\sigma,
\end{equation}
where $\mathcal{B}^\delta(\rho)$ is a ball of radius $\delta$ around $\rho_{XC}$. This ball is defined as the set of \emph{subnormalized} states $\sigma$ with $P(\rho,\sigma) 
\leq \delta$, where $P(\cdot,\cdot)$ is the \emph{purified distance}~\cite{tomamichel2010entropyduality}. The exact definition of the purified distance is not needed in this paper, so we omit it for simplicity and refer the interested reader to \cite{tomamichel2010entropyduality}. The only property of the purified distance that we need in this work is that for any (subnormalized) $\rho$ and $\sigma$,
\[ P(\rho,\sigma) \geq \frac{1}{2} \|\rho-\sigma\|.\]
This means that if $H_{\text{min}}^\delta(X|C)_\rho \geq k$, then there exists a subnormalized $\sigma_{XC}$ such that $\frac{1}{2}\|\rho-\sigma\| \leq \delta$ and $H_{\text{min}}(X|C)_\sigma \geq k$.

We can now state our main lemma. This can be generalised to the multi-source case in a straightforward manner.

\begin{lemma}\label{lem:smooth-entropy_bound}
Let $\mathrm{Ext} : \{0,1\}^{n_1} \times  \{0,1\}^{n_2} \to \{0,1\}^m$ be a $(k_1-\log 1/\eps_1-1,k_2-\log 1/\eps_2 - 1,\eps)$ quantum-proof two-source extractor in the Markov model. Then for any Markov state $\rho_{X_1X_2C}$ with $\HminSmooth[\rho]{\delta_1}{X_1|C} \geq k_1$ and
  $\HminSmooth[\rho]{\delta_2}{X_2|C} \geq k_2$,
\[\frac{1}{2} \norm{\rho_{\mathrm{Ext}(X_1,X_2)C} - \rho_{U_m} \otimes \rho_C} \leq 6 \delta_1 + 6 \delta_2 + 2\eps_1 + 2\eps_2 + 2\eps\]
if the extractor is weak, and 
\[\frac{1}{2} \norm{\rho_{\mathrm{Ext}(X_1,X_2)X_iC} - \rho_{U_m} \otimes \rho_{X_iC}} \leq 6 \delta_1 + 6 \delta_2 + 2\eps_1 + 2\eps_2 + 2\eps\] if the extractor is strong in the source $X_i$.
\end{lemma}

To prove that Lemma~\ref{lem:smooth-entropy_bound} holds, we first need to prove that if a state $\rho_{X_1X_2C}$ is guaranteed to be a Markov state with bounded smooth min-entropy, then there is a (subnormalized) state $\sigma_{X_1X_2C}$ close by which is also a Markov state with a bound on the min-entropy. This can be seen as a robustness property of the Markov model for extractors.

\begin{lemma}\label{lem:subnormalized_markov_chain}
  Let $\rho_{X_1X_2C}$ be a Markov state $X_1 \leftrightarrow C \leftrightarrow X_2$ such that
  $\HminSmooth[\rho]{\delta_1}{X_1|C} \geq k_1$ and
  $\HminSmooth[\rho]{\delta_2}{X_2|C} \geq k_2$. Then there exists a subnormalized
  state $\sigma_{X_1X_2C}$ such that $X_1,X_2$ and $C$ still form a Markov chain $X_1 \leftrightarrow C \leftrightarrow X_2$, and
  $\Hmin[\sigma]{X_1|C} \geq k_1 - \log \frac{1}{\eps_1}$,
  $\Hmin[\sigma]{X_2|C} \geq k_2 - \log \frac{1}{\eps_2}$ and
  $\frac{1}{2}\trnorm{\rho-\sigma} \leq \eps_1+\eps_2+3\delta_1+3\delta_2$.
\end{lemma}

\begin{proof}
  By the Markov chain condition, the state $\rho_{X_1X_2C}$ can
  equivalently be written
  \[ \rho_{X_1C_1ZE_2X_2} = \sum_{x_1,x_2,z} p(z) p(x_1|z) p(x_2|z) \proj{x_1}
  \otimes \rho^{x_1,z}_{C_1} \otimes \proj{z} \otimes \rho^{x_2,z}_{C_2}
  \otimes \proj{x_2}.\]
  Thus
  $\HminSmooth[\rho]{\delta_1}{X_1|C} =
  \HminSmooth[\rho]{\delta_1}{X_1|C_1Z}$
  and
  $\HminSmooth[\rho]{\delta_2}{X_2|C} =
  \HminSmooth[\rho]{\delta_2}{X_2|C_2Z}$. In the following we use only
  this form with the explicit classical register $Z$.

  By the definition of smooth min-entropy, we know that there exist
  (subnormalized) states
  \begin{align*} \tilde{\sigma}_{X_1C_1Z} & = \sum_{x_1,z} q_1(z) q(x_1|z) \proj{x_1} \otimes
  \sigma^{x_1,z}_{C_1} \otimes \proj{z} \\
  \text{and} \qquad 
  \hat{\sigma}_{X_2C_2Z} & = \sum_{x_2,z} q_2(z) q(x_2|z) \proj{x_2} \otimes
  \sigma^{x_2,z}_{C_2} \otimes \proj{z} \end{align*}
  such that $\frac{1}{2}\trnorm{\rho_{X_1C_1Z}-\tilde{\sigma}_{X_1C_1Z}} \leq \delta_1$,
  $\frac{1}{2}\trnorm{\rho_{X_2C_2Z}-\hat{\sigma}_{X_2C_2Z}} \leq \delta_2$,
  $\Hmin[\tilde{\sigma}]{X_1|C_1Z} \geq k_1$ and
  $\Hmin[\hat{\sigma}]{X_2|C_2Z} \geq k_2$.

Since $2^{-\Hmin[\sigma]{X_1|CZ}} = \sum_z q(z) 2^{-\Hmin[\sigma]{X_1|CZ=z}}$ also for subnormalized distributions $q(\cdot)$, we can define $2^{-\Hmin[\sigma]{X_1|CZ=z}} := 0$ when $q(z) = 0$, then pad $q(\cdot)$ to get a normalized distribution for which $2^{-\Hmin[\sigma]{X_1|CZ}} = \mathop{\mathbb{E}}_z \left[2^{-\Hmin[\sigma]{X_1|CZ=z}}\right]$. We can thus use Markov's inequality and get
\begin{align*} 
	\Pr_{z \leftarrow Z} \left[ \Hmin[\tilde{\sigma}]{X_1|C_1Z = z} \leq k_1 - \log \frac{1}{\eps_1} \right] & \leq \eps_1 \\
 	\text{and} \qquad \Pr_{z \leftarrow Z} \left[ \Hmin[\hat{\sigma}]{X_2|C_2Z =z} \leq k_2 - \log \frac{1}{\eps_2} \right] & \leq \eps_2 \;. 
\end{align*}

Let $\mathcal{Z}_1$ and $\mathcal{Z}_2$ be the sets of values for
which $q_1(z_1) \neq 0$, $q_2(z_2) \neq 0$, and
\begin{align*} 
	\forall z_1 \in \mathcal{Z}_1, \quad & \Hmin[\tilde{\sigma}]{X_1|C_1Z = z_1} \geq k_1 - \log \frac{1}{\eps_1} \\
	\text{and} \qquad \forall z_2 \in \mathcal{Z}_2, \quad & \Hmin[\hat{\sigma}]{X_2|C_2Z = z_2} \geq k_2 - \log \frac{1}{\eps_2} \;. 
\end{align*} 
Let
  $\bar{\mathcal{Z}} := \mathcal{Z}_1 \cap \mathcal{Z}_2$ be their intersection, and
  let $\bar{p}(z)$ be a subnormalized distribution given by
\[ \bar{p}(z) := \begin{cases} p(z) & \text{if
    $z \in \bar{\mathcal{Z}}$,} \\ 0 & \text{otherwise.} \end{cases} \]
We define the (subnormalized) state
\[ \sigma_{X_1C_1ZC_2X_2} := \sum_{x,y,z} \bar{p}(z) q(x_1|z) q(x_2|z)
\proj{x_1} \otimes \sigma^{x_1,z}_{C_1} \otimes \proj{z} \otimes
\sigma^{x_2,z}_{C_2} \otimes \proj{x_2},\]
and prove in the following that it satisfies the conditions of the
lemma.

By construction of $\sigma$ we have $2^{-\Hmin[\sigma]{X_1|C_1Z}} = \sum_z \bar{p}(z) 2^{-\Hmin[\sigma]{X_1|C_1Z=z}}$ for values $z$ such that $\Hmin[\sigma]{X_1|C_1Z=z} \geq k_1 - \log \frac{1}{\eps_1}$. Hence $\Hmin[\sigma]{X_1|C_1Z} \geq k_1 - \log \frac{1}{\eps_1}$ and similarly $\Hmin[\sigma]{X_2|C_2Z} \geq k_2 - \log \frac{1}{\eps_2}$.

 To bound the distance from $\rho_{X_1C_1ZC_2X_2}$, first note that
\[ \frac{1}{2}\sum_z \left|\bar{p}(z) - p(z) \right| \leq
\eps_1+\eps_2+\delta_1+\delta_2.\]
We also have
\begin{multline*}
  \sum_{x_1,z} \frac{p(z)}{2} \trnorm{p(x_1|z) \rho^{x_1,z}_{C_1} - q(x_1|z)
    \sigma^{x_1,z}_{C_1}}
  \leq \\ \sum_{x_1,z} \frac{1}{2}\trnorm{ p(z) p(x_1|z)
    \rho^{x_1,z}_{C_1} - q_1(z) q(x_1|z)
    \sigma^{x_1,z}_{C_1}} + \frac{1}{2}\trnorm{ q_1(z) q(x_1|z)
    \sigma^{x_1,z}_{C_1} - p(z) q(x_1|z)
    \sigma^{x_1,z}_{C_1}  } \leq 2 \delta_1.
\end{multline*}
The same holds for $X_2C_2Z$, namely 
\[ \sum_{x_2,z} \frac{p(z)}{2} \trnorm{p(x_2|z) \rho^{x_2,z}_{C_2} - q(x_2|z)
    \sigma^{x_2,z}_{C_2}} \leq 2 \delta_2.\]
 Putting this together we get
\begin{align*}
  & \frac{1}{2}\trnorm{\rho_{X_1C_1ZC_2X_2} - \sigma_{X_1C_1ZC_2X_2}} \\
  & \qquad = \sum_{x,y,z} \frac{1}{2}\trnorm{p(z) p(x_1|z) p(x_2|z)
    \rho^{x_1,z}_{C_1} \otimes \rho^{x_2,z}_{C_2} - \bar{p}(z) q(x_1|z)
    q(x_2|z) \sigma^{x_1,z}_{C_1} \otimes
    \sigma^{x_2,z}_{C_2}} \\
  & \qquad = \sum_{x,y,z} \frac{p(z)}{2}\trnorm{p(x_1|z) p(x_2|z)
    \rho^{x_1,z}_{C_1} \otimes \rho^{x_2,z}_{C_2} - q(x_1|z)
    q(x_2|z) \sigma^{x_1,z}_{C_1} \otimes
    \sigma^{x_2,z}_{C_2}} +\eps_1+\eps_2+\delta_1+\delta_2\\
  & \qquad \leq \sum_{x,y,z} \frac{p(z)}{2} \trnorm{p(x_1|z) p(x_2|z)
    \rho^{x_1,z}_{C_1} \otimes \rho^{x_2,z}_{C_2} - q(x_1|z)
    p(x_2|z) \sigma^{x_1,z}_{C_1} \otimes
    \rho^{x_2,z}_{C_2}} \\
  & \qquad \qquad \qquad {}+ \frac{p(z)}{2} \trnorm{q(x_1|z) p(x_2|z)
    \sigma^{x_1,z}_{C_1} \otimes \rho^{x_2,z}_{C_2} - q(x_1|z)
    q(x_2|z) \sigma^{x_1,z}_{C_1} \otimes
    \sigma^{x_2,z}_{C_2}} +\eps_1+\eps_2+\delta_1+\delta_2\\
  & \qquad = \sum_{x_1,z} \frac{p(z)}{2} \trnorm{p(x_1|z)
    \rho^{x_1,z}_{C_1} - q(x_1|z) \sigma^{x_1,z}_{C_1}} 
   + \sum_{x_2,z} \frac{p(z)}{2} \trnorm{p(x_2|z) \rho^{x_2,z}_{C_2} - 
    q(x_2|z) \sigma^{x_2,z}_{C_2}} +\eps_1+\eps_2+\delta_1+\delta_2 \\
& \qquad \leq \eps_1+\eps_2+3\delta_1+3\delta_2.\qedhere
\end{align*}
\end{proof}

Since Lemma~\ref{lem:subnormalized_markov_chain} finds a subnormalized state that is close, the next step is to prove that one can extract from subnormalized states. This is done in Appendix~\ref{sec:extracting_subnormalized} in Lemma~\ref{lem:extracting_subnormalized}. Combining this with a simple use of the triangle inequality allows us to prove Lemma~\ref{lem:smooth-entropy_bound}.

\begin{proof}[Proof of Lemma~\ref{lem:smooth-entropy_bound}]
We prove the case of a weak extractor $\mathrm{Ext}$. The proof for a strong extractor is identical.

By Lemma~\ref{lem:subnormalized_markov_chain} there exists a subnormalized Markov state $\sigma_{X_1X_2C}$ such that $\frac{1}{2}\norm{\rho_{X_1X_2C} - \sigma_{X_1X_2C}} \leq \eps_1 + \eps_2 + 3\delta_1 + 3\delta_2$ and $\Hmin[\sigma]{X_i|C} \geq k_i - \log 1/\eps_i$. Hence
\begin{align*}
& \frac{1}{2} \norm{\rho_{\mathrm{Ext}(X_1,X_2)C} - \rho_{U_m} \otimes \rho_C} \\
& \qquad \qquad \leq \frac{1}{2} \norm{\rho_{\mathrm{Ext}(X_1,X_2)C} - \sigma_{\mathrm{Ext}(X_1,X_2)C}} + \frac{1}{2} \norm{\sigma_{\mathrm{Ext}(X_1,X_2)C} - \rho_{U_m} \otimes \sigma_C} + \frac{1}{2} \norm{\rho_{U_m} \otimes \sigma_C - \rho_{U_m} \otimes \rho_C} \\
& \qquad \qquad \leq 2\eps_1 + 2\eps_2 + 6\delta_1 + 6\delta_2 + \frac{1}{2} \norm{\sigma_{\mathrm{Ext}(X_1,X_2)C} - \rho_{U_m} \otimes \sigma_C} \\
& \qquad \qquad \leq 2\eps_1 + 2\eps_2 + 6\delta_1 + 6\delta_2 + 2\eps\,,
\end{align*}
where in the last line we used Lemma~\ref{lem:extracting_subnormalized}.
\end{proof}

\subsection{Non-Markov sources}\label{sec:nonmarkov_sources}

It is trivial to show that if part of the side information $E$ is deleted, this cannot decrease the security of an extractor. As already observed in \cite{chung2014multi}, in the case of an extractor that is strong in the source $X_i$, any operation on $E$ conditioned on $X_i$ cannot help an adversary either. Intuitively, this holds because the adversary is given the entire source $X_i$, thus copying information about it to $E$ is pointless. We formalize this in the following lemma.

\begin{lemma}\label{lem:nonmarkov_source}
Let $\rho_{X_1EX_2}$ be a Markov source with $\Hmin[\rho]{X_i|E} \geq k_i$. Let $\mathcal{E} : \mathrm{L}(E) \to \mathrm{L}(E)$ be any CPTP map on $E$. If $\mathrm{Ext}$ is a $(k_1,k_2,\eps)$ quantum-proof two-source extractor, then it can be used to extract from $\sigma_{X_1EX_2} = \mathcal{E}(\rho_{X_1EX_2})$ with error $\eps$. Let $\mathcal{E} : \mathrm{L}(X_iE) \to \mathrm{L}(X_iE)$ be a CPTP map that leaves $X_i$ unmodified, i.e., $\mathcal{E}(\sum_x p_x \ketbra{x}{x} \otimes \rho^x_E) = \sum_x p_x \ketbra{x}{x} \otimes \mathcal{E}_x(\rho^x_E)$ for some set of CPTP maps $\mathcal{E}_x : \mathcal{L}(\mathcal{H}_E) \to \mathcal{L}(\mathcal{H}_E)$. If $\mathrm{Ext}$ is a $(k_1,k_2,\eps)$ quantum-proof two-source extractor strong in $X_i$, then it can be used to extract from $\sigma_{X_1EX_2} = \mathcal{E}(\rho_{X_1E_X2})$ with error $\eps$. \end{lemma}

\begin{proof}
We prove the case of the strong extractor. The proof for the weak extractor follows the same steps. We need to show that \[ \frac{1}{2} \trnorm{\sigma_{\mathrm{Ext(X_1,X_2)X_iE}} - \rho_{U_m} \otimes \sigma_{X_iE}} \leq \eps\,.\] This follows from the contractivity of the trace distance and because the maps $\mathrm{Ext}$ and $\mathcal{E}$ commute:
\begin{align*}
\frac{1}{2} \trnorm{\sigma_{\mathrm{Ext(X_1,X_2)X_iE}} - \rho_{U_m} \otimes \sigma_{X_iE}} & = \frac{1}{2} \trnorm{\mathcal{E}(\rho_{\mathrm{Ext(X_1,X_2)X_iE}}) - \rho_{U_m} \otimes \mathcal{E}(\sigma_{X_iE})} \\
& \leq \frac{1}{2} \trnorm{\rho_{\mathrm{Ext(X_1,X_2)X_iE}} - \rho_{U_m} \otimes \rho_{X_iE}} \leq \eps\,. \qedhere
\end{align*}
\end{proof}

An equivalent result in \cite[Theorem~4.1]{chung2014multi} allows the authors to prove that their complex information leaking model can be reduced to side information about one of the sources, which implies that a strong extractor in the Markov model is also an extractor in the model of \cite{chung2014multi}. Note that, as already observed in~\cite{chung2014multi}, the entropy of the state $\sigma_{X_1EX_2}$ is not meaningful, since the operation $\mathcal{E}$ might delete information without reducing the capacity to distinguish the output of the extractor from uniform. One has to measure the entropy on the Markov state before $\mathcal{E}$ is applied~\cite{chung2014multi}.

\section{Explicit constructions}\label{sec:specific_constructions}

In this section we give some examples for explicit constructions of quantum-proof multi-source extractors in the Markov model, as follows from our main theorem, Theorem~\ref{thm:quantum_proof_extractors}. 

In Section~\ref{sec:construction_1} we consider a two-source extractor by Dodis et al.~\cite{dodis2004improved}. This extractor requires the sum of the entropies in both sources to be larger than $n$, and we get a construction with nearly identical parameters in the quantum case. 
In Section~\ref{sec:construction_2} we consider a two-source extractor construction by Raz~\cite{raz05extractors}, which requires one source to have entropy at least $n/2$, whereas the other can be logarithmic. Here too, the resulting quantum-proof extractor has nearly identical parameters to the classical case. 
In Section~\ref{sec:construction_3} we look at a three source extractor by Li~\cite{li2015threesource}, which only requires the sources to have entropy poly-logarithmic in $n$. Plugging this in our main theorem allows a sublinear amount of entropy to be extracted in the quantum case, and by combining it with Trevisan's extractor~\cite{de2012trevisan}, we can extract the remaining entropy and thus obtain the same output length as in the classical case. 
The final construction we analyse in Section~\ref{sec:construction_4} is based on a recent two-source extractor by Li~\cite{li2015twosource}, which only needs two sources of poly-logarithmic min-entropy. Unfortunately, the error is $n^{-\Omega(1)}$, which means that Theorem~\ref{thm:quantum_proof_extractors} only allows $\Omega(\log n)$ bits to be extracted. Composing this with another variant of Trevisan's extractor~\cite{de2012trevisan} allows a sublinear amount of randomness to be extracted at the cost of requiring one of the sources to have $k = n^\alpha$ bits of entropy for any constant $\alpha < 1$.

Since the works of Dodis et al.~\cite{dodis2004improved} and Raz~\cite{raz05extractors} provide the exact parameters for their extractors, we do the same here below in Sections \ref{sec:construction_1} and \ref{sec:construction_2}. In contrast, for the two extractors from~\cite{li2015threesource,li2015twosource} the exact parameters are unknown, as only the simplified $O$-notation form is given in the corresponding papers. For this reason the constructions in Sections \ref{sec:construction_3} and \ref{sec:construction_4} are also given in $O$-notation.

\subsection{High entropy sources}
\label{sec:construction_1}

The first extractor we consider is a strong two-source extractor from Dodis et al.~\cite{dodis2004improved}, which requires both sources together to have at least $n$ bits of entropy.

\begin{lemma}[\cite{dodis2004improved}]
\label{lem:dodis_extractor}
	For any $n_1=n_2=n$, $k_1$, $k_2$ and $m$ there exists an explicit function $\mathrm{Ext} : \{0,1\}^{n} \times \{0,1\}^{n} \to \{0,1\}^m$ that is a $(k_1,k_2,\varepsilon)$ two-source extractor, strong in both $X_1$ and in $X_2$ (separately), with $\varepsilon=2^{-(k_1+k_2+1-n-m)/2}$.  
\end{lemma}

To have an error $\eps < 1$, the total entropy must be $k_1 + k_2 > n-1$. The difference between $k_1+k_2$ and $n-1$ can either be extracted or used to decrease the error. Let $\ell + m = k_1+k_2+1-n$, then the error is $\eps = 2^{-\ell/2}$.

Plugging Lemma~\ref{lem:dodis_extractor} into Theorem~\ref{thm:quantum_proof_extractors} we get the following.
\begin{cor}
	For any $n_1=n_2=n$, $k'_1$, $k'_2$ and $m$
        there exists an explicit function $\mathrm{Ext} : \{0,1\}^{n} \times \{0,1\}^{n} \to \{0,1\}^m$ that is a $(k'_1,k'_2,\varepsilon)$ two-source extractor, strong in both sources (separately), with $\varepsilon'=\frac{\sqrt{3}}{2}2^{-(k'_1+k'_2+1-n-5m)/8}$.\end{cor}

\begin{proof} From Theorem~\ref{thm:quantum_proof_extractors} we have $k'_1 = k_1 + \log \frac{1}{\eps}$ and $k'_2 = k_2 + \log \frac{1}{\eps}$. Rewriting the error from Lemma~\ref{lem:dodis_extractor} in terms of $m$ we get $m = k_1 + k_2 + 1 - n -2 \log \frac{1}{\eps}$. Hence $m = k'_1 + k'_2 + 1 - n - 4 \log \frac{1}{\eps}$, so $\eps = 2^{-(k'_1+k'_2+1-n-m)/4}$. Plugging this in the error from Theorem~\ref{thm:quantum_proof_extractors}, namely $\eps' = \sqrt{3\eps2^m}/2$ finishes the proof.  \end{proof}

The parameters in the quantum case are very similar to the classical one. We still need $k'_1 + k'_2 > n-1$ and the difference can either be extracted or used to decrease the error. But this time for $\ell + \tilde{m} = k'_1+k'_2+1-n$ the extractor outputs $m=\tilde{m}/5$ bits with error $2^{-\ell/8}$.

Since the extractor is strong we can compose it with a quantum-proof seeded extractor, e.g., Trevisan's extractor~\cite{de2012trevisan}, to extractor more randomness from the sources---this procedure is explained in Appendix~\ref{sec:composing_extractors}. Here we use a variant of Trevisan's extractor with parameters given in Lemma~\ref{lem:trev_1} in Appendix~\ref{sec:composing_extractors}.

\begin{cor}\label{cor:DEOR_trevisan_composition}
	For any $n_1=n_2=n$, $k'_1$, $k'_2$, $\eps'$, $m''$, $\eps''$, such that
        \begin{align*} m & = \frac{k'_1+k'_2+1-n-8 \log(\sqrt{3}/2\eps')}{5} \geq d \;,\\
\max[k'_1,k'_2] & \geq m'' + 4 \log \frac{m''}{\eps''} + 6 \;,
\end{align*} where $d$ is the seed length needed by the extractor from Lemma~\ref{lem:trev_1},
there exists an explicit function $\mathrm{Ext} : \{0,1\}^{n} \times \{0,1\}^{n} \to \{0,1\}^{m+m''}$ that is a quantum-proof $(k'_1,k'_2,\eps'+\eps'')$ two-source extractor.\end{cor}

We remark that the construction of Dodis et al.~\cite{dodis2004improved} is based on universal hash functions. These are already known to be good quantum-proof seeded extractors~\cite{renner2005universally,tomamichel2011extractor,hayashi2013dualhashing}. Recently, Hayashi and Tsurumaru~\cite{hayashi2013dualhashing} proved that they are also good quantum-proof extractors if the seed is not uniform. Using some of our proof techniques, the result of Hayashi and Tsurumaru can be generalised to obtain a different proof that the construction of Dodis et al.\ is a two-source extractor in the Markov model. The resulting parameters are better than what we obtain here with the generic reduction from quantum-proof to classical extractors, since the Hayashi-Tsurumaru proof~\cite{hayashi2013dualhashing} does not have the $\sqrt{2^{m}}$ factor.

\subsection{One high and one logarithmic entropy source}
\label{sec:construction_2}

The following construction by Raz~\cite{raz05extractors} improves on Dodis et al.~\cite{dodis2004improved}. One of the sources still requires at least $n/2$ bits of entropy, but the other can be logarithmic.

\begin{lemma}[\protect{\cite[Theorem 1]{raz05extractors}}]
\label{lem:raz_extractor}
For any $n_1$, $n_2$, $k_1$, $k_2$, $m$, and any $0 < \delta < 1/2$, such that,
\begin{align*}
n_1 & \geq 6 \log n_1 + 2 \log n_2,\\
k_1 & \geq \left(\frac{1}{2}+\delta\right)n_1 + 3 \log n_1 + \log n_2, \\
k_2 & \geq 5 \log (n_1-k_1), \\
m & \leq \delta \min\left[\frac{n_1}{8},\frac{k_2}{40}\right]-1,
\end{align*}
there exists an explicit function $\mathrm{Ext} : \{0,1\}^{n_1} \times \{0,1\}^{n_2} \to \{0,1\}^m$ that is a $(k_1,k_2,\eps)$-two-source extractor strong in both inputs (separately) with $\eps = 2^{-3m/2}$.
\end{lemma}

Plugging this into Theorem~\ref{thm:quantum_proof_extractors} we get the following.
\begin{cor}
\label{cor:raz_extractor}
For any $n_1$, $n_2$, $k'_1$, $k'_2$, $m$, and $0 < \delta' < 19/32$, such that,
\begin{align*}
n_1 & \geq 6 \log n_1 + 2 \log n_2,\\
k'_1 & \geq \left(\frac{1}{2}+\delta'\right)n_1 + 3 \log n_1 + \log n_2, \\
k'_2 & \geq \frac{163}{32} \log \left(\left(1+\frac{3\delta'}{19}\right)n_1-k'_1\right), \\
m & \leq \frac{16\delta'}{19} \min\left[\frac{n_1}{8},\frac{4k'_2}{163}\right]-1,
\end{align*} there exists an explicit function $\mathrm{Ext} : \{0,1\}^{n_1} \times \{0,1\}^{n_2} \to \{0,1\}^m$ that is a quantum-proof $(k'_1,k'_2,\eps)$-two-source extractor strong in both inputs (separately) with $\eps' = \frac{\sqrt{3}}{2} 2^{-m/4}$.
\end{cor}

\begin{proof}
We need $k'_1 \geq k_1 + \log 1/\eps$, so we set
\[
k'_1 = k_1 + \frac{3}{2}\delta \frac{n_1}{8} 
\geq \left(\frac{1}{2}+\frac{19\delta}{16}\right)n_1 + 3 \log n_1 + \log n_2. \]
We obtain the bound on $k'_1$ given above by setting $\delta' = 19\delta/16$. Similarly, we need 
$k'_2 \geq k_2 + \log 1/\eps$, so we set
\[
k'_2 = k_2 +\frac{3}{2} \frac{1}{2} \frac{k_2}{40} = \frac{163}{160} k_2 \geq \frac{163}{32} \log(n_1 - k_1).
\] Writing this in terms of $k'_1$ instead of $k_1$ gives the bound on $k'_2$. The bound on $m$ is also updated  in terms of $\delta'$ and $k'_2$. Finally the new error is given by $\eps' = \sqrt{3\eps2^m}/2$.
\end{proof}

Here too the parameters are very similar to the classical case, only the coefficients change somewhat.
As in Section~\ref{sec:construction_1}, this extractor is strong, hence we can compose it with Lemma~\ref{lem:trev_1} as explained in Appendix~\ref{sec:composing_extractors}.

\begin{cor}
For any $n_1$, $n_2$, $k'_1$, $k'_2$, $m$, and $0 < \delta' < 19/32$, satisfying the constraints from Corollary~\ref{cor:raz_extractor} and any $m''$, $\eps''$ such that
\begin{align*} m & \geq d(m'',\eps'') \;,\\
\max[k'_1,k'_2] & \geq m'' + 4 \log \frac{m''}{\eps''} + 6 \;,
\end{align*}  where $d$---the seed length needed by the extractor from Lemma~\ref{lem:trev_1}---is a function of $m''$ and $\eps''$,
there exists an explicit function $\mathrm{Ext} : \{0,1\}^{n} \times \{0,1\}^{n} \to \{0,1\}^{m+m''}$ that is a quantum-proof $(k'_1,k'_2,\frac{\sqrt{3}}{2} 2^{-m/4}+\eps'')$ two-source extractor.
\end{cor}

\subsection{Three poly-logarithmic sources}
\label{sec:construction_3}

The third extractor we consider can break the barrier of $n/2$ min-entropy---it is sufficient for the sources to have $k=\log^{12}n$ bits of entropy---but requires three sources instead of two.

\begin{lemma}[\protect{\cite[Theorem 1.5]{li2015threesource}}]\label{lem:li3}
For any $n$ and $k \geq \log^{12} n$, there exists an explicit function $\mathrm{Ext} : \{0,1\}^{n} \times \{0,1\}^{n} \times \{0,1\}^{n} \to \{0,1\}^m$ that is a $(k,k,k,\varepsilon)$ three-source extractor, strong in $X_1$ and in $X_2X_3$ with $m = 0.9k$ and $\varepsilon=2^{-k^{\Omega(1)}}$.  
\end{lemma}

Since the error of this extractor is not exponential in $k$, but only in $k^c$ for some $c < 1$, when applying it to a source with quantum side information we cannot extract all of the entropy, but only $k^{c'}$ bits, for any $c' < c$.

\begin{cor}\label{cor:li3}
For any $n$ and $k' \geq 2\log^{12} n$, there exists an explicit function $\mathrm{Ext} : \{0,1\}^{n} \times \{0,1\}^{n} \times \{0,1\}^{n} \to \{0,1\}^{m'}$ that is a quantum-proof $(k',k',k',\varepsilon')$ three-source extractor, strong in $X_1$ and in $X_2X_3$ with $m' = k'^{\Omega(1)}$ and $\varepsilon'=2^{-k'^{\Omega(1)}}$.  
\end{cor}

\begin{proof}
Let $c$ be the leading term in $\Omega(1)$ for $\varepsilon=2^{-k^{\Omega(1)}}$ from Lemma~\ref{lem:li3}. Note that we necessarily have $c < 1$, because otherwise for $k=n$ the error would be $2^{-n+o(n)}$ which is impossible~\cite{rad2000extractorbounds}. We thus get $k' = k + \log 1/\eps = k + k^c + o(k^c)$. Requiring $k' \geq 2\log^{12} n$ is sufficient to have $k \geq \log^{12} n$ for large enough $k$. Picking $m' = k^{c'} = k'^{\Omega(1)}$ for some $c' < c$ implies that $\eps' = \sqrt{4\eps2^m}/2 = 2^{-k^{\Omega(1)}} = 2^{-k'^{\Omega(1)}}$.
\end{proof}

Corollary~\ref{cor:li3} does not extract as much entropy as Lemma~\ref{lem:li3}, but it extracts enough to use as a seed in Trevisan's construction and extract the entropy of the sources $X_2X_3$. The parameters below are obtained by composing Corollary~\ref{cor:li3} with Lemma~\ref{lem:trev_2}.

\begin{cor}
There exists a constant $c'$ such that for any $n$ and $k_3 \geq k_2 \geq k_1 \geq \max \left[2\log^{12} n, \log^{3/c'}n \right]$, there exists an explicit function $\mathrm{Ext} : \{0,1\}^{n} \times \{0,1\}^{n} \times \{0,1\}^{n} \to \{0,1\}^m$ that is a quantum-proof $(k_1,k_2,k_3,\varepsilon)$ three-source extractor with $m = k_1^{\Omega(1)} +k_2+k_3 - o(k_2+k_3)$ and $\varepsilon=n^{-\Omega(1)}$.  
\end{cor}

\begin{proof}
The quantum-proof extractor from Lemma~\ref{lem:trev_2} requires a seed of length $d = O(\log^3 n)$ for an error $\eps = n^{-\Omega(1)}$. The output length of Corollary~\ref{cor:li3} is $m' = k_1^{c'} - o(k_1^{c'})$ for some constant $c'$. Thus, if $k_1^{c'} > \log^3 n$, the output is long enough.
\end{proof}

\subsection{Two poly-logarithmic sources}
\label{sec:construction_4}

In a recent breakthrough Chattopadhyay and Zuckerman constructed a two-source extractor that outputs $1$ bit and only requires two sources of poly-logarithmic entropy~\cite{chattopadhyay2015explicit}. This was then generalised to multiple output bits by Li~\cite{li2015twosource}.

\begin{lemma}[\protect{\cite[Theorem 1.3]{li2015twosource}}]\label{lem:li2}
There exists a constant $c_1$ such that for any $n$ and $k \geq \log^{c_1} n$, there exists an explicit function $\mathrm{Ext} : \{0,1\}^{n} \times \{0,1\}^{n} \to \{0,1\}^m$ that is a $(k,k,\eps)$ two-source extractor strong in $X_2$ with $m = k^{\Omega(1)}$ and $\eps=n^{-\Omega(1)}$.  
\end{lemma}

Since the error of this extractor is only polynomial in $1/n$, the quantum-proof version can only produce an output of length $m' = \Omega(\log n)$. The constant hidden in $m = \Omega(\log n)$ depends on the constant in $\eps=n^{-\Omega(1)}$. However, Lemma~\ref{lem:li2} allows the error to be $n^{-c_2}$ for any constant $c_2$~\cite{lipersonalcom}, which means that $m' = c_3 \log n$ for any $c_3$.

\begin{cor}
There exists a constant $c'_1$ such that for any $n$ and $k' \geq \log^{c'_1} n$, there exists an explicit function $\mathrm{Ext} : \{0,1\}^{n} \times \{0,1\}^{n} \to \{0,1\}^{m'}$ that is a quantum-proof $(k',k',\eps')$ two-source extractor strong in $X_2$ with $\eps'=n^{-\Omega(1)}$ and $m' = c_3 \log n$ for any constant $c_3 >0$ and sufficiently large $n$.
\end{cor}

\begin{proof}
Since for Lemma~\ref{lem:li2} we have $\eps = n^{-c_2}$ for any $c_2$, we set $m' = \frac{c_2}{2} \log n$, hence $\eps' = \sqrt{3\eps2^{m'}}/2 = n^{-\Omega(1)}$. The difference between $k'$ and $k$ is absorbed in the constant $c'_1$.
\end{proof}

This extractor is strong in the second input, hence as previously we can extract the entropy of this source using Trevisan's extractor. However, since the output is only $m' = c_3 \log n$, we compose it with a variant of Trevisan's extractor that only needs a seed of length $O(\log n)$, but requires the source to have entropy $k = n^\alpha$ for some constant $0 < \alpha \leq 1$ and extracts $k^{\beta}$ bits for any $0 < \beta < 1$. This extractor is given in Lemma~\ref{lem:trev_3}. The result given here below allows one of the sources to still have poly-logarithmic entropy, but requires the other to have $k = n^\alpha$ bits of min-entropy.

\begin{cor}
There exists a constant $c'_1$ such that for any $0 < \alpha \leq 1$, $0 < \beta < 1$, $n$, $k_1 \geq \log^{c'_1} n$ and $k_2 \geq n^\alpha$ there exists an explicit function $\mathrm{Ext} : \{0,1\}^{n} \times \{0,1\}^{n} \to \{0,1\}^{m''}$ that is a quantum-proof $(k_1,k_2,\eps'')$ two-source extractor with $\eps''=n^{-\Omega(1)}$ and $m'' = k_2^\beta$.
\end{cor}

\section{The quantum Markov model in quantum randomness amplification protocols}\label{sec:randomness_amplification}

Recently, the interest in quantum-proof two-source extractors (and multi-source in general) was renewed as people wished to use them as part of quantum randomness amplification (QRA) protocols. As for randomness extractors, the goal of a QRA protocol is to extract an almost uniformly random string from a weak source (which is usually known in public, e.g., NIST's Randomness Beacon). However, in contrast to randomness extractors, the idea is to do it with only \emph{one} weak source (and no seed) by exploiting the power of quantum physics (as mentioned in Section~\ref{sec:intro} this is impossible in the case of randomness extractors).  
Of course, once a quantum protocol is considered, it only makes sense to consider quantum side information.

With particular importance are QRA protocols which are device independent. That is, protocols in which one treats the devices as black boxes and does not assume much regarding the underlying quantum states and measurements inside the boxes\footnote{The advantage of this approach is that this stronger notion of security allows for some inevitable unknown imperfections in actual implementations of the protocol.}.
One should then prove the security of the protocol only based on the statistics which are observed by the honest user when running the protocol. This seemly impossible task is made possible by the use of Bell inequalities, which allow one to ``certify the quantumness'' of the considered protocol (for a review on the topic see, e.g.,~\cite{scarani2013device}). 

In the past couple of years several protocols were suggested for this task. The result presented in \cite{chung2014physical} was a big breakthrough: they considered a QRA protocol which uses a polynomial (at best) number of devices and a security proof against a general quantum adversary was proven. The main disadvantage of the protocol given in~\cite{chung2014physical} for actual implementations is the number of devices; each device can be thought of as a separate computer (or actually a complete laboratory where a Bell violation experiment can be done) and for the protocol to work one must make sure that the different computers cannot send signals to one another. Hence, a large number of devices amounts to a huge impractical apparatus. It is therefore interesting to ask whether a QRA protocol with a constant number of devices exists, or under which assumptions on the devices it is possible to devise such a protocol which can also be implemented in practice.  

Several other works considered the question of QRA with a constant number of devices, e.g.,~\cite{brandao2013robust,mironowicz2013amplification,plesch2014device}. The general idea in those works was to create two independent weak random sources from devices (under different additional setup assumptions not made in~\cite{chung2014physical}) that violate some Bell inequality, and then to apply a two-source extractor to get a final uniform key. 
However, as two-source extractors are not secure against general quantum adversaries the security was compromised. Indeed,~\cite{mironowicz2013amplification,plesch2014device} for example did not give a complete security proof against quantum adversaries. 
In~\cite{brandao2013robust} security was proven\footnote{The security proof of~\cite{brandao2013robust} holds against non-signalling adversaries, which are more powerful than quantum ones. Note that two-source extractors are not known to be secure against those more powerful non-signalling adversaries (in any model of the sources and the side information). Furthermore, our proof technique that shows security in specific quantum models cannot be used in the non-signalling case.} by a reduction to the case of a simple classical adversary (and hence the extractor could be used), at the cost of an additional setup assumption, namely that the adversary never has access to the initial weak source, and some loss in parameters.

Following the current work about quantum-proof multi-source extractors it is therefore interesting to consider the Markov model in the context of QRA protocols. More specifically, one can assume that two (or more) separated devices are a priori in product and become correlated only via the adversary or the environment, i.e., the state of the devices and the adversary $\rho_{ABC}$ is a Markov chain $A \leftrightarrow C \leftrightarrow B$. The (unknown but local) measurements in the two devices then create a ccq-state $\rho_{X_1X_2C}$ in the Markov model, to which the quantum-proof two-source extractor is applied.
 
Such assumptions about the structure of the devices could be justified in an intermediate device independent manner, e.g., if the devices are produced by two different experimental groups, or if the experimentalists know that a priori the devices are in a product state but might get correlated since they are placed in near by locations and therefore effected from the same temperature fluctuations. In any case, we still consider one quantum adversary and do not restrict her side information to some leakage operation as in~\cite{chung2014multi}.
Furthermore, it is well known that for many Bell inequalities, if the observed Bell violation in the QRA protocol is maximal then the devices must be in product with one another (i.e., one does not need to \emph{assume} that this is the case). Taking into account self-testing results like~\cite{reichardt2013classical,yang2014robust}, although out of reach of current techniques, it is possible that in the future one could justify an almost tensor product structure (in some appropriate notion of closeness under which the extractors still perform well) from a non-maximal observed Bell violation.

\section{Conclusions and open questions}\label{sec:open_questions}

In this work a new and natural model for classical and quantum-proof multi-source extractors was defined---the Markov model. We then showed that \emph{all} multi-source extractors, weak and strong, are also secure in the presence of side information that falls into the Markov model, both in the classical and quantum case. As explained in the previous sections, our main result, Theorem~\ref{thm:quantum_proof_extractors}, can be seen as a continuation, extension and improvement of previously known results~\cite{kasher2010two,chung2014multi}.

Apart from the result itself, on the technical level, a new proof technique was used, which, in contrast to the previous works is indifferent to whether the extractors are strong or not. In particular this implies that no adaptations are required for any new multi-source extractor that might be proposed in the future.

We finish this work with several open questions:

\begin{enumerate}
	\item Are there more general models that extend the Markov model in which all extractors remain secure? Some natural extensions were provided in Section~\ref{sec:extending_sources}.
        \item Are there different families of states $\rho_{X_1X_2C}$ from which it is possible to extract randomness that are relevant for practical applications? The difficulty in extracting randomness comes from the fact that we are not given one (known) state $\rho_{X_1X_2C}$, but that the extractor is expected to work for any state in a given family, e.g., a Markov state with lower bounds on the conditional min-entropy. The standard criterion of independence between the sources $X_1$ and $X_2$ has been relaxed in this work to allow for sources that are independent conditioned on $C$. Other structures might also allow randomness to be extracted.
	\item What happens if the sources and the side information are not exactly in the Markov model but only close to it? Even in the case of only two sources, there are different ways to quantify the closeness of a state to a Markov-chain state (see, e.g.,~\cite{fawzi2014quantum}). It is interesting to ask which notion of approximation is relevant in applications of multi-source extractors (such as quantum randomness amplification) and under which such notions the quantum-proof extractors remain secure. Note that the recovery map notion of approximation of Markov chains~\cite{fawzi2014quantum} does not guarantee approximate conditional independence of the sources, and seems to provide quite a different structure.
	\item It is unclear whether Theorem~\ref{thm:quantum_proof_extractors} is tight, i.e., whether the loss in the error of the extractor is inevitable when considering arbitrary extractors. In other words, it is not known if there are multi-source extractors for which the $\sqrt{2^m}$ loss in the error term is necessary\footnote{The same question arises in the case of seeded extractors as well~\cite{berta2014quantum,quantumbilinear}.}. In the other direction, as already noted in Section~\ref{sec:construction_1}, the work of~\cite{hayashi2013dualhashing} can be used, in combination with our proof technique, to show that for two-universal hashing (when the seed is taken to be the second source) the blow-up in the error term is not necessary.
	\item Do multi-source extractors remain secure also in the presence of non-signalling side information? Non-signalling adversaries are in general more powerful than quantum ones. For seeded extractors this does not seem to be the case~\cite{hanggi2009impossibility,arnon2012limits} but for multi-source extractors nothing is known. Note however that our proof technique is not applicable to non-signalling side information. 
\end{enumerate}

\paragraph{Acknowledgments.} We thank Mario Berta, Omar Fawzi and Thomas Vidick for helpful comments and discussions. This project was supported by the European Research Council (ERC) via grant No.\@ 258932, by the Swiss National Science Foundation (SNSF) via the National Centre of Competence in Research ``QSIT'', by the European Commission via the project ``RAQUEL'', and by the US Air Force Office of Scientific Research (AFOSR) via grant~FA9550-16-1-0245. VBS was also supported by the EC through grants ERC QUTE (NR 197868). The majority of this work was carried out while VBS was at ETH.

\appendix
\appendixpage

\section{Proofs of Section \ref{sec:security_proof}}\label{sec:proofs_security_two}

\begin{customlemma}{\ref{lem:classical_to_quantum_state}}
	Let $\rho_{X_1X_2C} = \sum_{x_1,x_2} \proj{x_1} \otimes \proj{x_2} \rho_C(x_1,x_2)$. Then there exists a pure state $\rho_{ABC}$ and POVMs $\{F_{x_1}\},\{G_{x_2}\}$ such that
	\begin{equation*}
 		\rho_C(x_1,x_2) = \tr_{AB}\left[ F_{x_1}^{\frac{1}{2}} \otimes G_{x_2}^{\frac{1}{2}} \otimes \mathbb{1}_{C}\rho_{ABC} F_{x_1}^{\frac{1}{2}} \otimes G_{x_2}^{\frac{1}{2}}  \otimes \mathbb{1}_{C}\right] \;.
	\end{equation*}
\end{customlemma}
\begin{proof}
	Let $\rho_{X_1X_2 C}= \sum_{x_1,x_2} p_{x_1,x_2} \kettbra{x_1}_{X_1} \otimes \kettbra{x_2}_{X_2} \otimes \tilde{\rho}^{x_1,x_2}_{C}$, where $\tilde{\rho}^{x_1,x_2}_{C} = \frac{\rho_C(x_1,x_2)}{\tr \rho_C(x_1,x_2)}$. And let $\ket{\psi^{x_1,x_2}}_{RC}$ be a purification of $\tilde{\rho}^{x_1,x_2}_{C}$. We define 
	\[	
		\ket{\rho}_{ABC} = \sum_{x_1,x_2} \sqrt{p_{x_1,x_2}} \ket{x_1}_{X_1} \otimes \ket{x_2}_{X_2} \otimes \ket{\psi^{x_1,x_2}}_{RC}
	\]
	 with $A=X_1$ and $B = X_2R$. One can easily verify that this lemma holds for $F_{x_1} = \kettbra{x_1}$ and $G_{x_2} = \kettbra{x_2} \otimes \mathbb{1}_{R}$.
\end{proof}

\begin{customlemma}{\ref{lem:simple_bound_weak}}
	For any $\Psi_{ABC_1C_2}$ and positive operators $\{F_{x_1}\},\{G_{x_2}\},\{H_{z_1}\},\{K_{z_2}\}$ which sum up to the identity,
	\begin{align*}
		\sum_{\substack{x_1,x_2,\\z_1,z_2,y}} 
		\left[ \delta_{\mathrm{Ext}\left(x_1,x_2\right)=y} - \frac{1}{M}\right]  
		\left[ \delta_{\mathrm{Ext}\left(z_1,z_2\right)=y} - \frac{1}{M}\right]  \; 
		\tr\left[ \Psi_{ABC_1C_2} F_{x_1}\otimes G_{x_2} \otimes H_{z_1}\otimes K_{z_2} \right] \\
		= \sum_{\substack{x_1,x_2,z_1,z_2 | \\ \mathrm{Ext}\left(x_1,x_2\right)=\mathrm{Ext}\left(z_1,z_2\right)}} \tr\left[ \Psi_{ABC_1C_2} F_{x_1}\otimes G_{x_2} \otimes H_{z_1}\otimes K_{z_2} \right] - \frac{1}{M}
	\end{align*}
\end{customlemma}

\begin{proof}
	Let $\mathrm{p}(x_1,x_2,z_1,z_2) = \tr\left[\Psi_{ABC_1C_2} F_{x_1}\otimes G_{x_2} \otimes H_{z_1}\otimes K_{z_2} \right]$.
        We consider each of the terms  of the LHS of the equation separately:
	\begin{align*}
		&\sum_{\substack{x_1,x_2,\\z_1,z_2,y}}\delta_{\mathrm{Ext}\left(x_1,x_2\right)=y}\delta_{\mathrm{Ext}\left(z_1,z_2\right)=y}\; 
		\mathrm{p}(x_1,x_2,z_1,z_2) =  \sum_{\substack{x_1,x_2,z_1,z_2 | \\ \mathrm{Ext}\left(x_1,x_2\right)=\mathrm{Ext}\left(z_1,z_2\right)}} \mathrm{p}(x_1,x_2,z_1,z_2) \;; \\
		&\frac{1}{M} \sum_{\substack{x_1,x_2,\\z_1,z_2,y}}\delta_{\mathrm{Ext}\left(x_1,x_2\right)=y}\; \mathrm{p}(x_1,x_2,z_1,z_2) =  \frac{1}{M} \;; \\
		& \frac{1}{M} \sum_{\substack{x_1,x_2,\\z_1,z_2,y}}\delta_{\mathrm{Ext}\left(z_1,z_2\right)=y}\; \mathrm{p}(x_1,x_2,z_1,z_2) = \frac{1}{M} \;; \\
		&\frac{1}{M^2}\sum_{\substack{x_1,x_2,\\z_1,z_2,y}} \mathrm{p}(x_1,x_2,z_1,z_2) = \frac{1}{M} \;.
	\end{align*}
	The lemma follows by combining all the terms.
\end{proof}

In the following we denote $[l]\setminus\{i\}$ by $\bar{i}$ and $[l]\setminus\{i,j\}$ by $\overline{\{i,j\}}$.

\begin{lemma}
\label{lem:multi.markov}
Let $\rho_{A_{[l]}C}$ be such that for all $i \in [l]$,
\[ I(A_i:A_{\bar{i}}|C) = 0\;.\] Then $\rho_{A_{[l]}C}$ can be written as a direct sum of product states, \[\rho_{A_{[l]}C} = \bigoplus_t \mathrm{p}(t) \rho^t_{A_1C_1^t} \otimes \dotsb \otimes \rho^t_{A_lC_l^t}\;,\] where $\mathcal{H}_C = \bigoplus_t \mathcal{H}_{C^t_1} \otimes \dotsb \otimes \mathcal{H}_{C^t_l}$.\end{lemma}

\begin{proof}
  We prove this lemmas by recursively applying the result from \cite{hayden2004structure} given in Equation~\eqref{eq:quantummarkov} on the structure of quantum Markov chains. We will also use the following facts about conditional mutual information:
\begin{enumerate}
	\item For any $\rho_{ABC}$, $I(A:B|C) \geq 0$.
	\item For any $\rho_{ABCD}$, $I(A:BC|D) \geq I(A:B|D)$.
	\item For any $\rho_{ABCX} = \sum_x p_x \rho^x_{ABC} \otimes \ket{x}\bra{x}$ classical on $X$, $I(A:B|CX) = \sum_x p_x I(A:B|CX=x)$.
\end{enumerate}

 Because $I(A_1:A_{\bar{1}}|C) = 0$, we know that
 \[ \rho_{A_{[l]}C} = \bigoplus_{t_1} p_{t_1} \rho^{t_1}_{A_1C_1^{t_1}}
   \otimes \rho^{t_1}_{A_{\bar{1}} C_{\bar{1}}^{t_1}}\;.\]
   Let $T_1$ denote a classical system defined by
   \[ \rho_{A_{[l]}CT_1} = \bigoplus_{t_1} p_{t_1}
   \rho^{t_1}_{A_1C_1^{t_1}} \otimes \rho^{t_1}_{A_{\bar{1}} C_{\bar{1}}^{t_1}} \otimes \ket{t_1}\bra{t_1}\;.\]
     Note that $\rho_{A_{[l]}CT_1}$ is related to $\rho_{A_{[l]}C}$ by
     an isometry from $C$ to $CT_1$, hence
     \[I(A_2:A_{\bar{2}}|CT_1) = I(A_2:A_{\bar{2}}|C) = 0 \;.\]
     It follows that for all $t_1$,
     \[I(A_2:A_{\bar{2}}|CT_1 = t_1) = 0 \;,\] and hence
\[I(A_2:A_{\overline{\{1,2\}}}|CT_1=t_1) = 0 \;,\] which means that
the state $\rho^{t_1}_{A_2 A_{\overline{\{1,2\}}}C^{t_1}}$ is a Markov chain
$A_2 \leftrightarrow C^{t_1} \leftrightarrow A_{\overline{\{1,2\}}}$. Applying Equation~\eqref{eq:quantummarkov} again, we get
\[ \rho^{t_1}_{A_2 A_{\overline{\{1,2\}}}C^{t_1}} = \bigoplus_{t_2} p_{t_2} \rho^{t_1,t_2}_{A_2C_2^{t_1,t_2}} \otimes
\rho^{t_1,t_2}_{A_{\overline{\{1,2\}}} C_{\overline{\{1,2\}}}^{t_1,t_2}}\;.\] Repeating this for all $i \in [l]$ proves the lemma.  \end{proof}

\section{Strong extractors}\label{sec:strong_proofs}
In this section we give the proofs necessary for the security of quantum-proof two-source extractors, \emph{strong} in the source $X_1$, against product side information. The same steps can be repeated to prove the same result for multi-source extractors which are strong with respect to other sources. 

The following lemma is the analogues of Lemma~\ref{lem:CS_weak} for the strong case.  

\begin{lemma}\label{lem:CS_strong}
Let $\rho_{X_1X_2C} = \rho_{X_1C_1} \otimes \rho_{X_2C_2}$ be a product ccq-state. Then there exists a POVM \(\{G_{z_2}\}\) acting on \(C_2\) such that
  \begin{multline*}
     \frac{1}{M}\| \rho_{\mathrm{Ext}(X_1,X_2)X_1C} - \rho_{U_m} \otimes \rho_{X_1C}\|^2 \leq\\
     \sum_{\substack{x_1,x_2,\\z_2,y}} \left[ \delta_{\mathrm{Ext}\left(x_1,x_2\right)=y} - \frac{1}{M}\right]   \left[ \delta_{\mathrm{Ext}\left(x_1,z_2\right)=y} - \frac{1}{M}\right]  \; \P[X_1 = x_1] \;\tr_{C_2}\left[\rho_{C_2}(x_2) G_{z_2} \right]\,,
  \end{multline*}
  where  $M=2^m$.
\end{lemma} 

\begin{proof}
  First, recall that for a hermitian matrix \(R\) we have \(\trnorm{R} = \max \{\tr[RS] \,:\,-\mathbb{1}\leq S \leq \mathbb{1}\}\). Applying this to the matrix which norm specifies the error of the extractor, we find
  \[
    \| \rho_{\mathrm{Ext}(X_1,X_2)C} - \rho_{U_m}\otimes \rho_{C}\| = \max_{-\mathbb{1}\leq S \leq \mathbb{1}}\tr\left[ \left(\rho_{\mathrm{Ext}(X_1,X_2)C} - \rho_{U_m}\otimes \rho_{C} \right) S\right] \,.
  \]
  Since $\rho_{\mathrm{Ext}(X_1,X_2)X_1C}$ and $\rho_{U_m}\otimes \rho_{X_1C}$ are block diagonal with respect to the outcome variable of the extractor \(y\), as well as to the classical variable \(x_1\), \(S\) can be assumed to be block diagonal as well. Using this and inserting the expression for \(\rho_{X_1X_2C}\) in Equation \eqref{eq:defomegacqstate} we arrive at
  \[
    \| \rho_{\mathrm{Ext}(X_1,X_2)X_1C} - \rho_{U_m}\otimes \rho_{X_1C}\| = \max_{-\mathbb{1} \leq S_{y,x_1} \leq \mathbb{1} }\,\sum_{y,x_1,x_2} \, \left[ \delta_{\mathrm{Ext}\left(x_1,x_2\right)=y} - \frac{1}{M}\right]\,\tr\left[\rho_{C_1}(x_1) \otimes \rho_{C_2}(x_2)  S_{x_1,y}\right]\,.
  \]
  Let us denote \(G_{x_2} = \bar{\rho}_{C_2}^{-\half} \rho_{C_2}(x_2)\bar{\rho}_{C_2}^{-\half}\) with \(\bar{\rho}_{C_2} = \sum_{x_2} \rho_{C_2}(x_2)\). Then we find
  \begin{align*}
    &\| \rho_{\mathrm{Ext}(X_1,X_2)X_1C} - \rho_{U_m}\otimes \rho_{X_1C}\| \\
    &\qquad= \max_{-\mathbb{1} \leq S_{y,x_1} \leq \mathbb{1} }\,\sum_{y,x_1,x_2} \, \left[ \delta_{\mathrm{Ext}\left(x_1,x_2\right)=y} - \frac{1}{M}\right]\,\tr\left[\left(\rho_{C_1}(x_1) \otimes \bar{\rho}_{C_2}\right)^\half \mathbb{1}_{C_1} \otimes G_{x_2} \left(\rho_{C_1}(x_1) \otimes \bar{\rho}_{C_2}\right)^\half S_{x_1,y}\right]\\
    &\qquad=\max_{-\mathbb{1} \leq S_{y,x_1} \leq \mathbb{1} }\,\sum_{y,x_1} \, \tr\left[\left(\rho_{C_1}(x_1) \otimes \bar{\rho}_{C_2}\right)^\half \mathbb{1}_{C_1} \otimes \Delta_{x_1,y} \left(\rho_{C_1}(x_1) \otimes \bar{\rho}_{C_2}\right)^\half S_{x_1,y}\right]
  \end{align*}
  where we used the abbreviation
  \[
    \Delta_{y,x_1} = \sum_{x_2} \left[ \delta_{\mathrm{Ext}\left(x_1,x_2\right)=y} - \frac{1}{M}\right] G_{x_2} \,.
  \]
  We now denote
  \begin{align*}
    \rho_{X_1C_1} = \sum_{x_1} \kettbra{x_1} \otimes \rho_{C_1}(x_1)
  \end{align*}
  and find \(\rho_{X_1C_1}^\half = \sum_{x_1} \kettbra{x_1} \otimes \rho_{C_1}(x_1)^\half\). Setting 
  \begin{align*}
    \Delta_y = \sum_{x_1} \kettbra{x_1} \otimes \mathbb{1}_{C_1} \otimes \Delta_{y,x_1}\,,\qquad S_y = \sum_{x_1} \kettbra{x_1} \otimes S_{y,x_1}
  \end{align*}
  we find
  \begin{align*}
    \| \rho_{\mathrm{Ext}(X_1,X_2)X_1C} - \rho_{U_m}\otimes \rho_{X_1C}\| = \max_{-\mathbb{1} \leq S_{y} \leq \mathbb{1} }\,\sum_{y} \,\tr\left[\rho_{X_1C_1}^\half \otimes \bar{\rho}_{C_2}^\half \Delta_y \rho_{X_1C_1}^\half \otimes \bar{\rho}_{C_2}^\half S_y\right]\,.
  \end{align*}
  The crucial observation is now that the sesquilinear form
  \begin{align*}
    (R_{y}) \times (T_{y}) \mapsto \sum_{y} \tr\left[\left(\rho_{X_1C_1} \otimes \bar{\rho}_{C_2}\right)^\half R^*_{y} \left(\rho_{X_1C_1} \otimes \bar{\rho}_{C_2}\right)^\half T_{y}\right]
  \end{align*}
  on block-diagonal matrices is positive semi-definite and hence fulfils the Cauchy-Schwarz inequality. Applying this gives
  \begin{align*}
    &\| \rho_{\mathrm{Ext}(X_1,X_2)X_1C} - \rho_{U_m}\otimes \rho_{X_1C}\|^2 \leq\\
    &\left(\sum_y  \tr\left[\left(\rho_{X_1C_1} \otimes \bar{\rho}_{C_2}\right)^\half \Delta_y \left(\rho_{X_1C_1} \otimes \bar{\rho}_{C_2}\right)^\half \Delta_y \right]\right) \cdot \left( \sum_y \tr\left[\left(\rho_{X_1C_1} \otimes \bar{\rho}_{C_2}\right)^\half S_y \left(\rho_{X_1C_1} \otimes \bar{\rho}_{C_2}\right)^\half S_y \right]\right) \,.
  \end{align*}
  Since we have that the norm of \(S_y\) is bounded by one, the terms in the second sum satisfy
  \begin{align*}
    \tr\left[\left(\rho_{X_1C_1} \otimes \bar{\rho}_{C_2}\right)^\half S_y \left(\rho_{X_1C_1} \otimes \bar{\rho}_{C_2}\right)^\half S_y \right] \leq \tr\left[\rho_{X_1C_1} \otimes \bar{\rho}_{C_2} S_y \right] \leq 1\,.
  \end{align*}
  Hence we arrive at
  \begin{align*}
    \| \rho_{\mathrm{Ext}(X_1,X_2)X_1C} - \rho_{U_m}\otimes \rho_{X_1C}\| &\leq \sqrt{M} \, \sum_y  \tr\left[\left(\rho_{X_1C_1} \otimes \bar{\rho}_{C_2}\right)^\half \Delta_y \left(\rho_{X_1C_1} \otimes \bar{\rho}_{C_2}\right)^\half \Delta_y \right]
  \end{align*}
  and expanding the definition of \(\Delta_y\) yields
  \begin{multline*}
    \| \rho_{\mathrm{Ext}(X_1,X_2)X_1C} - \rho_{U_m}\otimes \rho_{X_1C}\| \leq\\
     \sqrt{M}  \sum_{y,x_1,x_2,z_2}  \tr_{C_1}\left[\rho_{C_1}(x_1)\right] \left[ \delta_{\mathrm{Ext}\left(x_1,x_2\right)=y} - \frac{1}{M}\right] \left[ \delta_{\mathrm{Ext}\left(x_1,z_2\right)=y} - \frac{1}{M}\right] \tr\left[\rho_{C_2}(x_2)G_{z_2}\right]\,,
  \end{multline*}
  and since \(G_{z_2}\) are positive operators summing up to the identity, the assertion is proven.
\end{proof}

Next, let $\mathrm{p}(x_1,x_2,z_2)=\P[X_1 = x_1] \;\tr_{C_2}\left[\rho_{C_2}(x_2) G_{z_2} \right]$ and note that $\mathrm{p}(x_1,x_2,z_2)$ is indeed a probability distribution. Then, the following lemma is analogues to Lemma~\ref{lem:simple_bound_weak}. 

\begin{lemma}\label{lem:simple_bound_strong}
	For $\mathrm{p}(x_1,x_2,z_2)=\P[X_1 = x_1] \;\tr_{C_2}\left[\rho_{C_2}(x_2) G_{z_2} \right]$,
	\begin{equation}\label{eq:specific_attack_strong}
		\sum_{\substack{x_1,x_2,\\z_2,y}} \left[ \delta_{\mathrm{Ext}\left(x_1,x_2\right)=y} - \frac{1}{M}\right]   \left[ \delta_{\mathrm{Ext}\left(x_1,z_2\right)=y} - \frac{1}{M}\right]  \; \mathrm{p}(x_1,x_2,z_2) =
		 \sum_{\substack{x_1,x_2,z_2 | \\ \mathrm{Ext}\left(x_1,x_2\right)=\mathrm{Ext}\left(x_1,z_2\right)}} \mathrm{p}(x_1,x_2,z_2)  - \frac{1}{M}
	\end{equation}
\end{lemma}

\begin{proof}
        We follow a similar line as in the proof of Lemma~\ref{lem:simple_bound_weak}. 
	\begin{align*}
		&\sum_{\substack{x_1,x_2,\\z_2,y}}\delta_{\mathrm{Ext}\left(x_1,x_2\right)=y}\delta_{\mathrm{Ext}\left(x_1,z_2\right)=y}\; 
		\mathrm{p}(x_1,x_2,z_2) =  \sum_{\substack{x_1,x_2,z_2 | \\ \mathrm{Ext}\left(x_1,x_2\right)=\mathrm{Ext}\left(x_1,z_2\right)}} \mathrm{p}(x_1,x_2,z_2) \;; \\
		&\frac{1}{M} \sum_{\substack{x_1,x_2,\\z_2,y}}\delta_{\mathrm{Ext}\left(x_1,x_2\right)=y}\; \mathrm{p}(x_1,x_2,z_2) =  \frac{1}{M} \;; \\
		& \frac{1}{M} \sum_{\substack{x_1,x_2,\\z_2,y}}\delta_{\mathrm{Ext}\left(x_1,z_2\right)=y}\; \mathrm{p}(x_1,x_2,z_2) = \frac{1}{M} \;; \\
		&\frac{1}{M^2}\sum_{\substack{x_1,x_2,\\z_2,y}} \mathrm{p}(x_1,x_2,z_2) = \frac{1}{M} \;. \qedhere
	\end{align*}
\end{proof}

The quantity in Equation~\eqref{eq:specific_attack_strong} can be seen as a simple distinguishing strategy of a distinguisher trying to distinguish the output of the extractor from uniform given classical side information $Z_2$ about the second source $X_2$ and the source $X_1$. We can therefore relate it to the error of the \emph{strong} extractor in the case of classical side information, i.e., to Equation~\eqref{eq:classical_ext_error_strong}. This is shown in the following lemma, which is analogues to Lemma~\ref{lem:distinguishing_weak}. 

\begin{lemma}\label{lem:distinguishing_strong}
  Let $Z_2$ denote the classical side information about the source $X_2$.\footnote{There is no side information about the source $X_1$, since it is made available in full.}  Then
	\[
		\sum_{\substack{x_1,x_2,z_2 | \\ \mathrm{Ext}\left(x_1,x_2\right)=\mathrm{Ext}\left(x_1,z_2\right)}} \mathrm{p}(x_1,x_2,z_2) - \frac{1}{M} \leq \frac{1}{2}\| \mathrm{Ext}\left( X_1,X_2 \right)X_1Z_2 - U_m\circ X_1Z_2 \| \;.
	\]
\end{lemma}
\begin{proof}
	Define the following random variables over $\{0,1\}^m\times \{0,1\}^{n_1} \times \{0,1\}^{n_2}$: 
	\[
		\mathrm{R} = \mathrm{Ext}(X_1,X_2)X_1Z_2 \quad ; \quad \mathrm{Q} = U_m \circ X_1Z_2 \;.
	\]
	Let $\mathcal{A^\star} = \left\{ (a_1,a_2,a_3) \big| a_1 = \mathrm{Ext}\left(a_2,a_3\right) \right\} \subseteq \{0,1\}^m\times \{0,1\}^{n_1} \times \{0,1\}^{n_2}$. Then, the probabilities that $\mathrm{R}$ and $\mathrm{Q}$ assign to the event $\mathcal{A^\star}$ are 
	\[
		\mathrm{R}(\mathcal{A^\star}) = \sum_{\substack{x_1,x_2,z_2 | \\ \mathrm{Ext}\left(x_1,x_2\right)=\mathrm{Ext}\left(x_1,z_2\right)}} \mathrm{p}(x_1,x_2,z_2) \quad ; \quad \mathrm{Q}(\mathcal{A^\star}) = \frac{1}{M}
	\]
	Using the definition of the variational distance we therefore have 
	\begin{align*}
		\frac{1}{2} \| \mathrm{Ext}(X_1,X_2)X_1Z_2 - U_m \circ X_1Z_2\|  &= \sup_{\mathcal{A}} \| \mathrm{R}(\mathcal{A}) - \mathrm{Q}(\mathcal{A}) \| \\
		& \geq  \mathrm{R}(\mathcal{A^\star}) - \mathrm{Q}(\mathcal{A^\star}) \\
		& = \sum_{\substack{x_1,x_2,z_2 | \\ \mathrm{Ext}\left(x_1,x_2\right)=\mathrm{Ext}\left(x_1,z_2\right)}} \mathrm{p}(x_1,x_2,z_2) - \frac{1}{M} \;. \qedhere
	\end{align*}
\end{proof}

Finally, we combine the lemmas together to show that any strong classical-proof two-source extractor in the Markov model is secure against product quantum side information as well. We follow similar steps to those in the proof of Lemma~\ref{lem:classical_markov_to_product_quantum}. 

\begin{lemma}\label{lem:classical_markov_to_product_quantum_strong}
  Any $(k_1,k_2,\varepsilon)$-strong classical-proof two-source extractor in the Markov model is a $\left(k_1, k_2,\eps'\right)$-strong quantum-proof product two-source extractor with $\eps' = \sqrt{\varepsilon \cdot 2^{(m-2)}}$, where $m$ is the output length of the extractor.  \end{lemma}

\begin{proof}
	Let $\rho_{X_1X_2C}=\rho_{X_1C_1}\otimes\rho_{X_2C_2}$ be any state of two classical sources and product side information with $\Hmin{X_1|C_1} \geq k_1$ and $\Hmin{X_2|C_2} \geq k_2$.

	We can apply Lemmas  \ref{lem:CS_strong}, \ref{lem:simple_bound_strong}, and \ref{lem:distinguishing_strong} to get the bound
	\begin{equation}\label{eq:combined_bound_strong}
		\| \rho_{\mathrm{Ext}(X_1,X_2)X_1C} - \rho_{U_m} \otimes \rho_{X_1C}\| \leq \sqrt{ \frac{M}{2} \| \mathrm{Ext}\left( X_1,X_2 \right)X_1Z_2 - U_m\circ X_1Z_2 \|} \;.
	\end{equation}
	
	As it follows from the proofs of the previous lemmas that $Z_2$ includes side information about $X_2$ alone (and there is no additional side information about $X_2$, i.e., the quantum system $C_1$ is just thrown away) $\mathrm{p}(x_1,x_2,z_2)=\mathrm{p}(x_1)\cdot\mathrm{p}(x_2,z_2)$, which implies: 
	\begin{enumerate}
		\item The sources and the classical side information form a Markov chain $X_1 \leftrightarrow Z_2 \leftrightarrow X_2$.
		\item $H_{\text{min}}\left(X_1|Z_2\right) = H_{\text{min}}\left(X_1\right) \geq H_{\text{min}}\left(X_1|C_1\right)$. 
		\item$H_{\text{min}}\left(X_2|Z_2\right) \geq H_{\text{min}}\left(X_2|C_2\right)$.
	\end{enumerate}
	
	Hence, if $H_{\text{min}}\left(X_i|C_i\right)\geq k_i$ then by the definition of a strong classical-proof two-source extractor,
	\begin{equation}\label{eq:class_extractor_promise_strong}
		\frac{1}{2} \| \mathrm{Ext}\left( X_1,X_2 \right)X_1Z_2 - U_m\circ X_1Z_2 \| \leq \varepsilon \;. 
	\end{equation}
	
	Combining Equations \eqref{eq:combined_bound_strong} and \eqref{eq:class_extractor_promise_strong} we get 
	\[
		\frac{1}{2} \| \rho_{\mathrm{Ext}(X_1,X_2)X_1C} - \rho_{U_m} \otimes \rho_{X_1C}\| \leq \frac{1}{2} \sqrt{M\varepsilon} = \sqrt{\varepsilon 2^{(m-2)}} \;. \qedhere
	\]
\end{proof}

\section{Extracting from subnormalized states}
\label{sec:extracting_subnormalized}

Extractors are usually defined for normalized states $\rho_{X_1X_2C}$. In applications one might wish to extract from subnormalized states---for example, the smooth min-entropy of a state is a bound on the entropy of a subnormalized state that is close by. Here we prove that if a function is an extractor (for normalized states), then one can use it to extract from subnormalized states as well. We write up the lemma and proof in the case of two-source extractors in the Markov model. Similar statements hold for multiple sources as well as seeded extractors.

\begin{lemma}\label{lem:extracting_subnormalized}
  Let $\sigma_{X_1X_2C}$ be a subnormalized quantum Markov state satisfying $\Hmin[\sigma]{X_1|C} \geq k_1$ as well as $\Hmin[\sigma]{X_2|C} \geq k_2$, and let $\mathrm{Ext}: \{0,1\}^{n_1} \times \{0,1\}^{n_2} \to \{0,1\}^m$ be a $(k_1-1,k_2-1,\varepsilon)$ quantum-proof two-source extractor in the Markov model. If $\mathrm{Ext}$ is weak, then we have that 
\[
  	\frac{1}{2}\norm{\sigma_{\mathrm{Ext}(X_1,X_2)C} - \rho_{U_m} \otimes \sigma_C} \leq 2\eps\,.
\]
 If $\mathrm{Ext}$ is strong in $X_i$, then we have that 
\[
  	\frac{1}{2}\norm{\sigma_{\mathrm{Ext}(X_1,X_2)X_iC} - \rho_{U_m} \otimes \sigma_{X_iC}} \leq 2\eps\,.
\]
\end{lemma}

\begin{proof} We prove the weak case. The proof for strong extractors is identical.

	Define $p = \tr[\sigma_C]$ and with that the normalized state $\hat{\sigma}_{X_1X_2C} = \frac{1}{p}\sigma _{X_1X_2C}$ as well as the auxiliary normalized state
\[
		\tilde{\sigma}_{X_1X_2CP} = \sigma_{X_1X_2C} \otimes \proj{0}_P + (1-p) \tau_{X_1X_2} \otimes \hat{\sigma}_C \otimes \proj{1}_P,
\]
	where $\tau_{X_1X_2}$ is the fully mixed state. Note that $X_1 \leftrightarrow CP \leftrightarrow X_2$ is a Markov chain for the state $\tilde{\sigma}_{X_1X_2CP}$.
	This state satisfies slightly modified min-entropy conditions: 
\[
p_{\text{guess}}(X_1|CP)_{\tilde{\sigma}_{X_1CP}} =  p_{\text{guess}}(X_1|C)_{\sigma_{X_1C}}+(1-p)p_{\text{guess}}(X_1|C)_{\tau_{X_1}\otimes\hat{\sigma}_C} = 2^{-k_1} + (1-p)2^{-n_1} \leq 2 \cdot 2^{-k_1}\;.
\]
	Hence $\Hmin[\tilde{\sigma}]{X_1|CP} \geq k_1 - 1$, 
	and the same argument can also be carried out for $X_2$ showing that $\Hmin[\tilde{\sigma}]{X_2|CP} \geq k_2 - 1$. The state $\tilde{\sigma}_{X_1X_2CP}$ is thus a valid Markov state satisfying the min-entropy conditions and hence we have
\[
		\frac{1}{2}\norm{\tilde{\sigma}_{\mathrm{Ext}(X_1,X_2)CP} - \rho_{U_m} \otimes \tilde{\sigma}_{CP}} \leq \eps \,.
\]
	But since the partial trace over the $P$ system only decreases the trace distance, we infer that
\[
		\frac{1}{2}\norm{\sigma_{\mathrm{Ext}(X_1,X_2)C} - \rho_{U_m} \otimes \sigma_C + (1-p) \tau_{\mathrm{Ext}(X_1,X_2)} \otimes \hat{\sigma}_C - (1-p)\rho_{U_m} \otimes \hat{\sigma}_C} \leq \eps\,.
\]
	Thus starting from the expression $\norm{\sigma_{\mathrm{Ext}(X_1,X_2)C} - \rho_{U_m} \otimes \sigma_C}$ and then adding and subtracting the term $(1-p)[\tau_{\mathrm{Ext}(X_1,X_2)} \otimes \hat{\sigma}_C - \rho_{U_m} \otimes \hat{\sigma}_C]$ as well as applying the triangle inequality leaves us with
\[
		\frac{1}{2}\norm{\sigma_{\mathrm{Ext}(X_1,X_2)C} - \rho_{U_m} \otimes \sigma_C} \leq \eps + \frac{1-p}{2}\norm{\tau_{\mathrm{Ext}(X_1,X_2)} \otimes \hat{\sigma}_C - \rho_{U_m} \otimes \hat{\sigma}_C} \leq 2\eps \,,		
\]
	since $\tau_{X_1X_2} \otimes \hat{\sigma}_C$ is a Markov source satisfying the entropic constraints.
\end{proof}

\section{Composing two-source and seeded extractors}
\label{sec:composing_extractors}

If a multi-source extractor is strong in an input $X_1$, then the output $Y$ is independent from $X_1$. 
This can be interpreted as $Y$ containing the entropy from $X_2$; the randomness of $X_1$ served only as a catalyst, but is still contained in that random variable. A very common technique used to extract that randomness is to use another extractor. Since $Y$ is uniform and independent from $X_1$, it fulfils the conditions needed to use it as a seed in seeded extractor. This is formalised in the following lemma.

\begin{lemma}
\label{lem:composing_extractors}
Let $\mathrm{Ext} : \{0,1\}^{n_1} \times \{0,1\}^{n_2} \to \{0,1\}^d$ be a quantum-proof $(k_1,k_2,\eps)$-two-source extractor strong in the first input. And let $\mathrm{Ext'} : \{0,1\}^{n_1} \times \{0,1\}^{d} \to \{0,1\}^m$ be a quantum-proof $(k_1,\eps')$-seeded extractor. Then the function
\begin{align*}
\mathrm{Ext}'' : \ & \{0,1\}^{n_1} \times \{0,1\}^{n_2} \to \{0,1\}^m \\
& (x_1,x_2) \mapsto \mathrm{Ext}'(x_1,\mathrm{Ext}(x_1,x_2)),
\end{align*}
is a quantum-proof $(k_1,k_2,\eps+\eps')$-two-source extractor.
\end{lemma}

\begin{proof}
Let $\rho_{U_dX_1C} = \rho_{U_d} \otimes \rho_{X_1C}$, where $\rho_{U_d}$ is a fully mixed state of dimension $2^d$. And let $\rho_{\mathrm{Ext}'(X_1,U_d)C}$ denote the state resulting from applying $\mathrm{Ext}'$ to $X_1$ with $U_d$ as seed. From the triangle inequality and contractivity of the trace distance we have
\begin{align*}
& \frac{1}{2}\trnorm{\rho_{\mathrm{Ext}'(X_1,\mathrm{Ext}(X_1,X_2))C} - \rho_{U_m} \otimes \rho_C} \\
& \qquad \qquad \qquad \leq \frac{1}{2}\trnorm{\rho_{\mathrm{Ext}'(X_1,\mathrm{Ext}(X_1,X_2))C} -\rho_{\mathrm{Ext}'(X_1,U_d)C}} +\frac{1}{2}\trnorm{\rho_{\mathrm{Ext}'(X_1,U_d)C} - \rho_{U_m} \otimes \rho_C} \\
& \qquad \qquad \qquad \leq \frac{1}{2}\trnorm{\rho_{\mathrm{Ext}(X_1,X_2)X_1C} -\rho_{U_d} \otimes \rho_{X_1C}} +\frac{1}{2}\trnorm{\rho_{\mathrm{Ext}'(X_1,U_d)C} - \rho_{U_m} \otimes \rho_C}\;.
\end{align*}
The first term above is the error of $\mathrm{Ext}$ and the second is the error of $\mathrm{Ext}'$.
\end{proof}

Note that Lemma~\ref{lem:composing_extractors} only requires a weak seeded extractor. Hence if a strong extractor is used, the seed can additionally be appended to the output---this is the case for all the following extractors.

Here below we give several seeded quantum-proof extractor constructions---all variants of Trevisan's extractor---that we use in the explicit constructions from Section~\ref{sec:specific_constructions}.

The first construction~\cite[Corollary 5.3]{de2012trevisan} is one for which the exact parameters have been calculated~\cite{mauerer2012trevisan}.

\begin{lemma}[\protect{\cite[Corollary 5.3]{de2012trevisan},\cite{mauerer2012trevisan}}]
\label{lem:trev_1}
There exists an explicit function $\mathrm{Ext} : \{0,1\}^n \times \{0,1\}^d \to \{0,1\}^m$, which is a quantum-proof $(k,\eps)$-strong extractor with
\begin{align*}
  t & =  2 \log \frac{2nm^2}{\eps^2}\;, \\
  a & = 1 + \max\left\{ 0, \frac{\log (m-e) - \log (t-e)}{\log e - \log (e-1)}\right\}\;, \\
  k & = m+ 4 \log \frac{m}{\eps} + 6\;, \\
  d & = at^2\;,
\end{align*}
where $e$ is the mathematical constant.
\end{lemma}

The entropy loss of this extractor, $k-m = 4 \log \frac{m}{\eps} + 6$, can be reduced by composing it with an almost universal hash function~\cite{tomamichel2011extractor}.

\begin{lemma}[\protect{\cite[Corollary 5.4]{de2012trevisan}}]
\label{lem:trev_2}
There exists an explicit function $\mathrm{Ext} : \{0,1\}^n \times \{0,1\}^d \to \{0,1\}^m$, which is a quantum-proof $(k,\eps)$-strong extractor with
$k  = m + 4 \log \frac{1}{\eps} + O(1)$ and $d = O(\log^2\frac{n}{\eps}\log m)$.
\end{lemma}

For $\eps = n^{-\Omega(1)}$ in Lemma~\ref{lem:trev_2} we get $d = O(\log^3 n)$.

The final construction we consider only requires a seed of length $O(\log n)$, but can only extract a sublinear amount of entropy.
\begin{lemma}[\protect{\cite[Corollary 5.6]{de2012trevisan}}]
\label{lem:trev_3}
For any constant $0 < \gamma < 1$ there exists an explicit function $\mathrm{Ext} : \{0,1\}^n \times \{0,1\}^d \to \{0,1\}^m$, which is a quantum-proof $(k,\eps)$-strong extractor with
$k  = n^\gamma m + 8 \log \frac{m}{\eps} + O(1)$, $d = O(\log n)$ and $\eps = n^{-\Omega(1)}$.
\end{lemma}

For example, if $k = n^\alpha$ in Lemma~\ref{lem:trev_3} for $\gamma < \alpha \leq 1$, then $m = n^{\alpha - \gamma} - o(1) = k^{1 - \frac{\gamma}{\alpha}} - o(1)$.

\clearpage

\newcommand{\etalchar}[1]{$^{#1}$}

\end{document}